\documentclass[aoas]{imsart}

\RequirePackage{amsthm,amsmath,amsfonts,amssymb}
\RequirePackage[authoryear]{natbib}
\RequirePackage[colorlinks,citecolor=blue,urlcolor=blue]{hyperref}
\RequirePackage{graphicx}
\usepackage{tikz}
\startlocaldefs
\theoremstyle{plain}

\newtheorem{theorem}{Theorem}[section]

\newtheorem{corollary}[theorem]{Corollary}
\theoremstyle{definition}

\endlocaldefs

\begin{document}

\begin{frontmatter}
\title{Inhomogeneous continuous-time Markov chains to infer flexible time-varying evolutionary rates}
%\title{A sample article title with some additional note\thanksref{t1}}
\runtitle{Bayesian phylogenetic inference under time-varying evolutionary rates}
%\thankstext{T1}{A sample additional note to the title.}

\thankstext{corrauth}{\emph{Corresponding author.}  Email: \href{mailto:msuchard@ucla.edu}{msuchard@ucla.edu}}

\begin{aug}
%%%%%%%%%%%%%%%%%%%%%%%%%%%%%%%%%%%%%%%%%%%%%%%
%% Only one address is permitted per author. %%
%% Only division, organization and e-mail is %%
%% included in the address.                  %%
%% Additional information can be included in %%
\author[A]{\fnms{Pratyusa}~\snm{Datta} \ead[label=e1]{pratyusa@ucla.edu}\orcid{0009-0007-3733-920X}},
\author[B]{\fnms{Philippe}~\snm{Lemey}\ead[label=e2]{philippe.lemey@kuleuven.be}\orcid{0000-0003-2826-5353}}
\and
\author[A]{\fnms{Marc}~\snm{A. Suchard}\thanksref{corrauth}\ead[label=e3]{msuchard@ucla.edu}\orcid{0000-0001-9818-479X}}
%%%%%%%%%%%%%%%%%%%%%%%%%%%%%%%%%%%%%%%%%%%%%%
%% Addresses                                %%
%%%%%%%%%%%%%%%%%%%%%%%%%%%%%%%%%%%%%%%%%%%%%%
\address[A]{Department of Biostatistics,
University of California, Los Angeles, USA \printead[presep={ ,\ }]{e1,e3}}

\address[B]{Department of Microbiology, Immunology and Transplantation, Rega Institute, KU Leuven \printead[presep={,\ }]{e2}}

\end{aug}

\begin{abstract}
Reconstructing evolutionary histories and estimating the rate of evolution from molecular sequence data is of central importance in evolutionary biology and infectious disease research. 
We introduce a flexible Bayesian phylogenetic inference framework that accommodates 
changing evolutionary rates over time
by modeling sequence character substitution processes as inhomogeneous continuous-time Markov chains (ICTMCs) acting along the unknown phylogeny, where the rate remains as an unknown, positive and integrable function of time. 
The integral of the rate function appears in the finite-time transition probabilities  of the ICTMCs that must be efficiently computed for all branches of the phylogeny to evaluate the observed data likelihood. 
Circumventing computational challenges that arise from a fully nonparametric function, we successfully parameterize the rate function as piecewise constant with a large number of epochs that we call the polyepoch clock model. 
This makes the transition probability computation relatively inexpensive and continues to flexibly capture rate change over time. 
We employ a Gaussian Markov random field prior to achieve temporal smoothing of the estimated rate function.  
Hamiltonian Monte Carlo sampling enabled by scalable gradient evaluation under this model makes our framework computationally efficient. 
We assess the performance of the polyepoch clock model in recovering the true timescales and rates through simulations under two different evolutionary scenarios. 
We then apply the polyepoch clock model to examine the rates of West Nile virus, Dengue virus and influenza A/H3N2 evolution, and  
estimate the time-varying rate of SARS-CoV-2 spread in Europe in 2020. 
\end{abstract}


\end{frontmatter}
%%%%%%%%%%%%%%%%%%%%%%%%%%%%%%%%%%%%%%%%%%%%%%
%% Please use \tableofcontents for articles %%
%% with 50 pages and more                   %%
%%%%%%%%%%%%%%%%%%%%%%%%%%%%%%%%%%%%%%%%%%%%%%
%\tableofcontents

\section{Introduction}

Phylogenetics is central to understanding the unknown ancestral relationships relating biological entities and the mechanisms that drive their diversification over time.
Reconstructing these relationships into a phylogenetic tree describing this evolutionary history relies on the accumulation of genetic substitutions shared by descent within molecular sequences.
%An essential goal in evolutionary biology is to reconstruct the phylogenetic tree describing the evolutionary history from such molecular sequence data, 
%where either 
Fossil information \citep{YangRannala2005} or sampling molecular sequences at different time points \citep{Rambaut2000}  can help calibrate this unobserved phylogeny or tree with respect to calendar time, further allowing the estimation of divergence times of different lineages. 
Likelihood-based phylogenetic inference methods generally assume that the observed molecular sequences arise from a homogeneous continuous-time Markov chain (CTMC) acting along the branches of the phylogeny. 
Phylogenetic methods for divergence time estimation rescale this CTMC through
molecular clock models that parameterize the expected number of substitutions per sequence site per unit-time or evolutionary rate on each branch.
The strict molecular clock \citep{Zuckerland1965} is the simplest and assumes a shared rate across all branches. 
More relaxed clock models \citep{Yoder2000, Kishino2001, Yang2003, Drummond2006, Drummond2010} make contrasting assumptions about the patterns of variation in evolutionary rate across different branches or their resulting lineages.

\par 

Phylogenetics also plays a key role in infectious disease research by providing critical insight into how pathogenic viruses emerge and spread. 
Genome sequences of rapidly evolving viruses accumulate significant genetic change over 
time-scales often ranging from a few centuries to even weeks. 
This enables researchers to sample viral genome sequences in real-time over outbreaks \citep{dudas2016} and reconstruct the viral phylogeny from these heterochronous sequences  to estimate the origin of outbreaks, divergence times of specific viral lineages, and the rate at which the virus evolves.  
Evolutionary rate estimates, however, may vary substantially depending on the timeframe of sampling for several viruses 
\citep{Aiewsakun2016, Duchene2014}. 
The common observation is that the rates estimated over shorter time-scales, such as those calculated from contemporary outbreaks, are typically higher than the rates estimated over longer time-scales, such as those based on virus-host co-speciation history or geographic calibration. 
Several clock models have been proposed to correct for this dependence of evolutionary rate estimates on the time-frame of measurement. 
\cite{aiewsakun2015b} show that a simple power law rate decay helps explain the discrepancy between short-term and long-term rate estimates of the foamy virus. 
\cite{Membrebe2019} combine a small-scale change-point or epoch modeling approach \citep{Bielejec2014} with a power-law relationship between rate and time, and estimate a strong time-dependent effect for foamy virus and lentivirus, where the estimated rates vary over four orders of magnitude.

\par 

Dependence of the evolutionary rate on calendar time is evident from these studies, particularly for fast evolving viruses. 
This highlights the importance of accommodating more flexible modeling of the  evolutionary rate directly as a function of time in the phylogenetic inference framework
and requires a relaxation of the current homogeneous CTMC assumption along individual branches in the phylogeny.
In this paper, we overcome this limitation by exploring
inhomogeneous continuous-time Markov chains (ICTMCs) acting along the branches of the phylogeny in which 
the infinitesimal generator matrix characterizing the ICTMC becomes a function of time. 
In general, the finite-time transition probability matrix of an ICTMC has no closed-form expression \citep{Fortmann1977, Kailath1980}.
If we assume, however, that the time-varying evolutionary rate scales all elements of the infinitesimal generator equally, remains positive over all time and is an integrable function with respect to time, then  
the transition probabilities matrices do retain closed-form expressions \citep{Rindos}
involving these simpler integrals that we need to evaluate for all branches to return the observed molecular sequence likelihood.

In a Bayesian inference setting, positing a prior for this unknown evolutionary rate function stands paramount for both incorporating sufficient biological flexibility and yielding practical implementation.
While Gaussian processes \citep{Rasmussen2006} are typical choices for unknown functions, they currently lack efficient approximations of their integrals after necessary transformation to retain positivity.  
We overcome this limitation by 
proposing a solution to make the computation of the transition probabilities inexpensive as well as analytically tractable and flexible.
In particular, we model the evolutionary rate as 
a log-transformed (positive),
  piecewise constant function of time, with a large number of change-points, and we call this the \textit{polyepoch clock model}. 
  
For the polyepoch clock model, we further develop an effective inference framework  
to simultaneously estimate the phylogeny and evolutionary rate trajectory over time by employing a Gaussian Markov random field prior \citep{rueheld} to incorporate the temporal dependence between the large number of highly correlated rate parameters corresponding to the many epochs. 
Importantly, we propose a new algorithm  that computes the gradients of the log posterior with respect to all rate parameters in just linear-time as a function of the number of taxa. 
This enables us to use efficient gradient-based Hamiltonian Monte Carlo (HMC) approaches \citep{Neal2011} to sample from the high dimensional posterior density of the rate parameters. 
Through simulation, we assess the performance of the polyepoch clock model and apply it to examine the time-varying evolutionary rates of 
the West Nile virus in North America (1999 - 2007), 
serotype 3 of Dengue virus  in Brazil (1972 - 2010), 
seasonal influenza A/H3N2 across the globe and 
the rate of spatial diffusion of SARS-CoV2 in Europe from January to October 2010 \citep{Lemey2021}. 
We observe strong time-varying patterns in the estimated rate trajectory for Dengue, A/H3N2 and SARS-CoV-2.

\section{Methods}\label{sec:Methods}

\newcommand\data{\mathbf{Y}}
\newcommand{\datum}{Y}
\newcommand\columns{C}
\newcommand\column{c}
\newcommand\states{S}
\newcommand\state{s}
\newcommand\numtips{N}
\newcommand\branchLengths{T}
\newcommand\branchLength{b}
\newcommand\terminalBranches{\varepsilon}
\newcommand\internalBranches{\mathcal{I}}
\newcommand\phylogeny{{\cal F}}
\newcommand\branch{b}
\newcommand\rates{R}
\newcommand\rate{r}
\newcommand{\siteSpecificRate}{\gamma}
\newcommand\postorderPartial{\mathbf{p}}
\newcommand\postorderElement{p}
\newcommand\preorderPartial{\mathbf{q}}
\newcommand\preorderElement{q}
\newcommand{\probabilityMatrix}[2]{\mathbf{P}^{(#1)} \hspace{-0.2em} \left( #2 \right)}
\newcommand{\probabilityEntry}[4]{P^{(#1)}_{#3 #4} \hspace{-0.2em} \left( #2 \right)}
\newcommand\rateMatrix{\mathbf{Q}}
\newcommand\disjointSet{Y}
\newcommand{\below}[1]{\mathbf{Y}_{\lfloor #1 \rfloor}}
\newcommand{\abbove}[1]{\mathbf{Y}_{\lceil #1 \rceil}}
\newcommand{\transpose}{'}
\newcommand{\estate}{t}
\newcommand{\parent}{k}
\newcommand{\sibling}{j}
\newcommand{\probability}[1]{\mathbb P \hspace{-0.2em} \left( #1 \right)}
\newcommand{\rootDistribution}{\boldsymbol{\pi}}
\newcommand{\given}{\,|\,}
\newcommand{\pa}[1]{{\mathrm{pa}(#1)}}

Relating $N$ aligned molecular sequences lies an unobserved phylogenetic tree $\mathcal{F}$ with $N$ tip nodes and $N - 1$ internal nodes. 
We denote the tip nodes with $1, \dots, N$ and the internal nodes with $N + 1, \dots, 2N - 1$, where we denote the root node with $2N - 1$. 
The sequences are observed at the tip or external nodes only and remain unobserved at the internal nodes.  
Any branch on $\mathcal{F}$ connects a parent node to its child node where the parent node is closer to the root. 
We denote $\pa{i}$ as the parent of node $i$. 
We denote the real time of node $i$ with $t_i$ and the length of the branch connecting node $i$ to $\pa{i}$ with $b_i$, where branch length is measured in expected number of substitutions along the branch per site.
We refer to a branch by the number of the child node it connects.
We model the sequence alignment as arising from conditionally independent ICTMCs acting along the branches of $\mathcal{F}$.
We denote $\cal{S}$ as the state-space of the ICTMC and $S$ as the number of possible states in $\cal{S}$ (e.g. $S = 4$ for nucleotide substitution models).
We denote $C$ as the number of aligned sites.
Suppose we have observed (at tips) and latent (at internal nodes) discrete evolutionary characters $Y_{ic} \in S$, for $i = 1, \dots, 2N - 1$, and $c = 1, \dots, C$. 
We denote the observed data as $\mathbf{Y} = (\mathbf{Y}_1, \dots, \mathbf{Y}_{N})'$, where the columns $\mathbf{Y}_i = (Y_{i1}, \dots, Y_{iC})$ are conditionally independent, for $i  = 1, \dots, 2N - 1$, often with across-site rate variation modeled with a discretized-Gamma prior on site-specific rates, but more complex processes can also be considered \citep{Gill:2025aa}.

We denote the inifnitesimal generator matrix characterizing the ICTMC by $\mathbf{Q}(t)$, where the off-diagonal elements of $\mathbf{Q}(t)$ are the instantaneous transition rates between two different states in $S$ at time $t$ and the diagonals are such that each row of $\mathbf{Q}(t)$ sums to $0$. 
We denote the finite-time transition probability matrix of the ICTMC from time $t_0$ to $t$ by $\mathbf{P}(t_0, t)$.
If $\mathbf{Q}(t)$ is piecewise continuous, then the general solution for $\mathbf{P}(t_0, t)$ is given by the Peano-Baker series 
\citep{Fortmann1977, Kailath1980}:
\begin{equation}\label{eq:peano_baker}
\begin{split}
\mathbf{P}(t_0, t) &= \mathbf{I} + \int_{t_0}^{t} \mathbf{Q}(\tau_1) d\tau_1 + \int_{t_0}^{t} \mathbf{Q}(\tau_1) d\tau_1 \int_{t_0}^{\tau_1} \mathbf{Q}(\tau_2) d\tau_2 d\tau_1 + \dots \\
& + \int_{t_0}^{t} \mathbf{Q}(\tau_1) \int_{t_0}^{\tau_1} \dots \int_{t_0}^{\tau_{k - 1}} \mathbf{Q}(\tau_k) d\tau_k d\tau_{k - 1} \dots d\tau_2 d\tau_1 + \dots . 
\end{split}
\end{equation}
This series does not have a closed-form expression in general. 
However if we assume that $\mathbf{Q}(t)$ and $\int_{t_0}^{t} \mathbf{Q}(\tau) d\tau$ commute for all $t$, then the transition probability matrix has a further simplified expression which is given by
\begin{equation}\label{eq:pwc_tpm}
\mathbf{P}(t_0, t) = \exp \left[\int_{t_0}^{t} \mathbf{Q}(\tau) d\tau \right].
\end{equation}
The commutative property between $\mathbf{Q}(t)$ and $\int_{t_0}^{t} \mathbf{Q}(\tau) d\tau$ holds for all $t$ when we assume $\mathbf{Q}(t) = f(t) \mathbf{Q}$ for $t > 0$, where $\mathbf{Q}$ is a valid infinitesimal generator matrix independent of time and $f(t)$ is an unknown, positive and integrable function describing the evolutionary rate trajectory over time. 
Under this assumption, the transition probability matrix of the ICTMC acting along the branch connecting node $i$ to $\pa{i}$ includes the integral of $f(t)$ from $t_{\pa{i}}$ to $t_{i}$ \citep{Rindos} as shown below
 \begin{equation}
 \mathbf{P}(t_{\pa{i}}, t_i) = \exp\left[\left(\int_{t_{\pa{i}}}^{t_i} f(t) dt \right) \mathbf{Q}\right]. 
 \label{eq:ictmc-tpm}
\end{equation}

To evaluate the data likelihood, the integral $ \int_{t_{\pa{i}}}^{t_i} f(t) d t$ needs to be computed for all $2N - 2$ branches of the phylogeny.
To ensure that $f(t)$ remains flexible and inexpensively evaluable, we model
$f(t)$ as a positive, piecewise constant function of time, with a large number of blocks or epochs, 
as the polyepoch clock model.
We assume that there are $M$
temporal grid-points denoted as $ w_M < \dots < w_1$ that divide time into $M + 1$ epochs where $w_1$ is closest to the present,
and model  
$f(t) = \theta_{M + 1}$ for $t < w_{M}$ and $f(t) = \theta_m$, $w_m \leq t < w_{m-1}$, $m = 1, \dots, M$, where $w_0$ is the present. 
Let $\boldsymbol{\theta} = (\theta_1, \dots, \theta_{M + 1})$ denotes the vector of rate parameters where $\theta_m > 0$ for $m = 1,  \dots, M + 1$. 
We use the example phylogenetic tree in Figure \ref{phylo-tree-1} with 3 tip nodes and 2 internal nodes to demonstrate the polyepoch clock model with $2$ grid-points $w_2 < w_1$ (gray dashed lines) 
and rate parameters $\theta_1, \theta_2, \theta_3$.
The real time of each node is depicted with black dotted lines. The rate along the branch connecting node $1$ to its parent, node $4$, is $\theta_2$ from $t_4$ to $w_1$ and $\theta_1$ from $w_1$ to $t_1$. Similarly the rate varies in a piecewise constant manner along the other three branches as well.
\begin{figure}[!htp]
	\centering
	\includegraphics[width=0.8\linewidth]{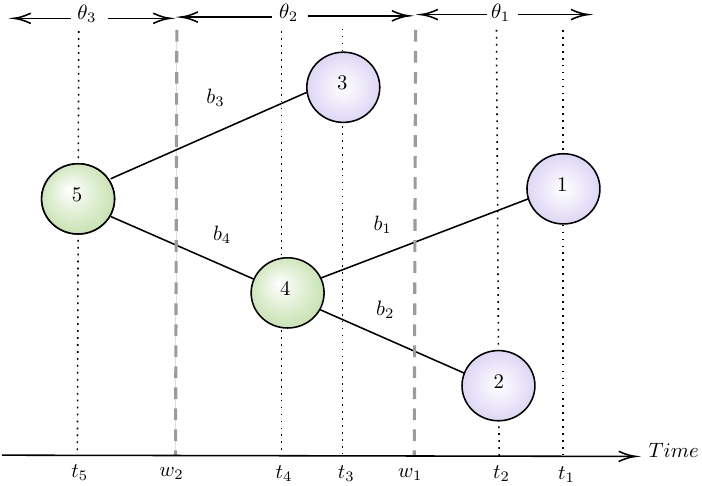} 
\caption{Piecewise constant evolutionary rate along a phylogenetic tree with three taxa. 
The real time of each node is depicted with black dotted lines.
Time is divided into three epochs by the two grid-points $w_2 < w_1$ (gray dashed lines) such that the rate is $\theta_3$ at any time point less than $w_2$, $\theta_2$ at any time point greater than equal to $w_2$ and less than  $w_1$, and $\theta_1$ at any time point greater than or equal to $w_1$.}
	\label{phylo-tree-1}
\end{figure}

In general, if $w_p$ is the largest grid-point less than or equal to $t_i$ and $w_q$ is the smallest grid-point greater than or equal to $t_{\pa{i}}$, then the integral appearing in $\mathbf{P}(t_{\pa{i}}, t_i)$, which is essentially the expected number of substitutions $b_i$ from $t_{\pa{i}}$ to $t_i$, reduces to
\begin{equation}
b_i = \int_{t_{\pa{i}}}^{t_i} f(t) dt =  \left( \theta_{q + 1} (w_q - t_{\pa{i}} ) + \sum_{m = p} ^ {q - 1} \theta_{m + 1} (w_m - w_{m +1}) + \theta_p (t_i - w_p)  \right).
\label{eq:pctdm-tpm}
\end{equation}
Given the transition probabilities, we can proceed to evaluate the observed data likelihood under the polyepoch clock model.
To ease presentation, we focus on the data likelihood for a single site. 
Given the rate parameters $\boldsymbol{\theta}$, the phylogeny $\mathcal{F}$ and other evolutionary parameters, the likelihood at a given site is the probability of the characters observed at the tips, marginalized over all possible latent states at the internal nodes.
The likelihood can be computed in a tractable manner using the pruning algorithm \citep{Felsenstein1973, Felsenstein1981}, which follows a post-order tree traversal, starting from the tips, visiting each node once in a descendant to parent fashion, until the root node is reached. 
We note that the observed data likelihood for multiple sites under the across-site rate variation model is a simple weighted-mixture model over these likelihoods.
Apart from the rate parameters $\boldsymbol{\theta}$ and the phylogeny $\mathcal{F}$, other evolutionary parameters that appear in the likelihood include elements of $\mathbf{Q}$ and the across site rate variation model parameter, commonly denoted with $\alpha$.

\subsection{Bayesian Inference}\label{sec:bayesian_inference}

Bayesian phylogenetic inference proceeds by sampling from the joint posterior density of all parameters of interest given the observed data, which is proportional to the product of the observed data likelihood and the joint prior density.
First, we need to choose a suitable prior for our primary parameters of interest $\boldsymbol{\theta}$ such that the prior incorporates the temporal dependence between the rates in adjacent epochs.
A prior popularly used in the literature to model dependent random variables is the Gaussian Markov random field \citep{rueheld}. 
We reparameterize $\zeta_m = \log \theta_m$ for $m = 1, \dots, M + 1$ and model the first-order difference between $\zeta_1, \dots, \zeta_{M + 1}$ to be distributed as follows
\begin{equation}
\zeta_{m+1}-\zeta_m \given \tau \overset{\text{ind}}{\sim} N(0, d_m/\tau), \quad m = 1, \ldots, M.
\label{eq:weighted_gmrf1}
\end{equation}
Here $d_{m}$ is the absolute difference between the midpoints of the $m$th and $(m+1)$th epochs. 
This assigns an improper Gaussian Markov random field (GMRF) prior to $\boldsymbol{\zeta} = (\zeta_1, \dots, \zeta_{M + 1})'$
\begin{equation}
\mathrm{P}(\boldsymbol{\zeta} \given \tau) \propto \tau^{\frac{M}{2}}\exp{-\left[\frac{\tau}{2} \boldsymbol{\zeta}'\left(\mathbf{D_w} - \mathbf{W}\right) \boldsymbol{\zeta}\right]},
\label{eq:weighted_gmrf2}
\end{equation}
where the elements of the $(M+1) \times (M + 1)$ matrix $\mathbf{W} = \left\{ W_{jk} \right\}$ are defined as $W_{jk} = 0$ for $| k - j| \neq 1$ and $W_{jk}= d_{\min\{j, k\}}$ for $|k - j| = 1$. 
The $(M+1) \times (M + 1)$ diagonal matrix $\mathbf{D_w}$ is defined as $\mathbf{D_w} = \text{diag}(w_{1+}, \dots, w_{(M+1)+})$ where $w_{j+}$ corresponds to the sum of the elements of the $j$th row of $\mathbf{W}$ for $j = 1, \dots, M + 1$. 
The prior \eqref{eq:weighted_gmrf2} corresponds to a singular multivariate Gaussian distribution since each row of $\mathbf{D_w} - \mathbf{W}$ sums to 0. 
When using improper priors, it is necessary to ensure that the resulting posterior distributions are proper in order to sample from them using MCMC algorithms. 
We provide a justification on why the resulting posterior distribution of $\boldsymbol{\theta}$ remains improper if we use the improper GMRF prior for $\boldsymbol{\zeta}$  in Appendix~\ref{sec:improperGMRF_appendix}. 
We proceed therefore with a class of proper GMRF priors which replaces $\mathbf{D_w} - \mathbf{W}$ in \eqref{eq:weighted_gmrf2} with $\mathbf{D_w} - \rho \mathbf{W}$ as follows 
\begin{equation}
\mathrm{P}(\boldsymbol{\zeta} \given \tau) \propto \tau^{\frac{M}{2}}\exp{-\left[\frac{\tau}{2} \boldsymbol{\zeta}'\left(\mathbf{D_w} - \rho \mathbf{W}\right) \boldsymbol{\zeta}\right]}.
\label{eq:propprior}
\end{equation}
For $|\rho| < 1$, $\mathbf{D_w} - \rho \mathbf{W}$ is non-singular due to diagonal dominance, which asserts that, for any symmetric matrix $\mathbf{A}$,  if $a_{jj} > 0$ and $a_{jj} > \sum_{k \neq j} |a_{jk}|$, $\mathbf{A}$ is positive definite. %PL: delete 'then'?
Here $\rho$ can be thought of as the expected proportional reaction of $\log \theta_m$ to the log transformed rates in its adjacent epochs. 
Since the prior is improper for $\rho = 1$ and proper for $\rho < 1$, choosing say $\rho = 0.99$, allows very large temporal association between the rates across successive epochs while ensuring that the posterior is proper.
Alternatively, we can assign a suitable prior such as the logit normal prior to $\rho$.

\par 

Further, the GMRF precision $\tau$ controls the smoothness of the rate trajectory. 
Since we do not usually have prior knowledge about the smoothness, we assign the following relatively uninformative Gamma prior to $\tau$
\begin{equation}
\mathrm{P}(\tau) \propto \tau^{a - 1} e^{- b \tau},
\label{precisionprior}
\end{equation}
where $a = b = 0.001$ so that the prior has expectation $1$ and variance $1000$. 
We assess the sensitivity of the polyepoch clock model to the choice of diffuse priors for $\tau$ in Appendix~\ref{sec:sensitivity_appendix}.

\par

In addition to priors for $\boldsymbol{\theta}$, $\tau$ and for $\rho$ if not fixed, we assign suitable priors $\mathrm{P}(\mathbf{Q})$, $\mathrm{P}(\alpha)$ and $\mathrm{P}(\mathcal{F})$ to the $\mathbf{Q}$, $\alpha$ and $\mathcal{F}$ respectively following standard Bayesian phylogenetic practices \citep{Baele2025}.
We sample from the following joint posterior density 
\begin{equation}
\mathrm{P}(\boldsymbol{\theta}, \tau, \rho, \mathbf{Q}, \alpha, \mathcal{F} \given \mathbf{Y}) \propto \mathrm{P}(\mathbf{Y} \given \boldsymbol{\theta}, \mathbf{Q}, \alpha, \mathcal{F}) \mathrm{P}(\boldsymbol{\theta}\given\tau, \rho) \mathrm{P}(\tau) \mathrm{P}(\rho) \mathrm{P}(\mathbf{Q}) \mathrm{P}(\alpha) \mathrm{P}(\mathcal{F}).
\label{eq:jointposterior}
\end{equation}
To estimate the high dimensional posterior density of the rate parameters that are highly correlated, we use HMC to sample from their full conditional distribution within a Metropolis-within-Gibbs scheme. 
HMC is a gradient-based Markov chain Monte Carlo (MCMC)
method that exploits numerical solutions of Hamiltonian dynamics \citep{Neal2011} to efficiently explore the parameter space and scales well with higher dimensions. 
The partial derivatives of the log posterior distribution with respect to  $\theta_m$ required for HMC sampling can be computed as follows for $m = 1, \dots, M + 1$:
 \begin{equation}
 	\frac{\partial}{\partial \theta_{m}} \log \mathrm{P}(\boldsymbol{\theta}, \tau, \rho, \mathbf{Q}, \alpha, \mathcal{F} \given \mathbf{Y}) = \sum_{i=1}^{2N-2}
 	\frac{\partial \log \mathrm{P}(\boldsymbol{\theta}, \tau, \rho, \mathbf{Q}, \alpha, \mathcal{F} \given \mathbf{Y})}{\partial b_{i}}  \frac{\partial b_{i}}{\partial \theta_{m}}.
 	\label{eq:gradient_rate} 
 \end{equation}
The gradient $\partial \log \mathrm{P}(\boldsymbol{\theta}, \tau, \rho, \mathbf{Q}, \alpha, \mathcal{F} \given \mathbf{Y}) / \partial b_{i}$ for all $2N - 2$ branches can be obtained simultaneously in linear-time complexity ${\cal O}(N C S^2)$ \citep{Ji2019}.
 The gradient $\partial b_i / \partial \theta_m$ of the $i$-th branch length with respect to mth rate parameter can be computed from equation~\ref{eq:pctdm-tpm} using the fact that $b_i$ is a linear combination of $\theta_1, \dots, \theta_{M + 1}$ for $i = 1, \dots, 2N - 2$.
 However, the contribution of $\theta_m$ to $b_i$ is positive only if the branch connecting node $i$ to $\pa{i}$ cuts through the mth epoch $(w_{m - 1}, w_m]$ and $0$ otherwise.
 The cost of computing $\partial b_i / \partial \theta_m$ for all branches and rates is most expensive, i.e. $\mathcal{O}(NM)$, when each branch cuts through each of the $M + 1$ epochs, and this cost can therefore be less than $\mathcal{O}(NM)$ in general, depending on the tree and the grid-points.
Therefore, this chain-rule extension for all $\theta_m$ simultaneously carries a cost of atmost $\mathcal{O}(N C S^2 + NM)$. % is the NK smaller due to the tree?
Calculating the gradient $\partial \log \mathrm{P}(\boldsymbol{\theta}, \tau, \rho, \mathbf{Q}, \alpha, \mathcal{F} \given \mathbf{Y}) / \partial b_{i}$ for all $2N - 2$ branches using the standard peeling algorithm would incur a higher cost of  $\mathcal{O}(N^2 C S^2)$ operations, leading to an overall gradient evaluation cost of atmost $\mathcal{O}(N^2 C S^2 + NM)$ operations. Finite difference approximations of the gradients, such as central differences, would require $\mathcal{O}(NCMS^2)$ operations, which becomes prohibitive for large $M$, as often required for accurate temporal resolution. Thus, our analytic, linear-time likelihood gradient evaluation combined with the chain rule offers the most computationally efficient strategy.

 \section{Results}
\label{sec:Results}
 
 \subsection{Simulation Study}\label{sec:simulation}
 
We examine the ability of the polyepoch clock model to recover the true evolutionary rate function and time of most recent common ancestor (tMRCA) under two different simulation scenarios. 
For each scenario, we first infer the maximum clade credibility tree (MCC) obtained by analyzing the Dengue virus example from \cite{Nunes2014} while assuming a random-effects relaxed clock model \citep{Ji2019}, a non-parametric coalescent prior \citep{Skyride} and the general-time-reversible substitution process \citep{Tavare1986}. 
For the first simulation, we generate simulated sequence data on the tips of this MCC tree using $\pi$BUSS \citep{pibuss} assuming a strict clock model and the HKY substitution process \citep{hky} with a fixed transition/transversion ratio of 2. 
We 
then draw inference both
under a strick clock and the polyepoch clock model with $100$ epochs of equal width between 2005 and 1963.
We fix the last grid-point 
so that 
it
covers
the expected tree height and captures all the information about the evolutionary rate 
from the data. 
For inference, we assume 
the HKY substitution process \citep{hky} and a non-parametric tree prior \citep{Skygrid}. 
The top panel of Figure~\ref{fig:scrsim} shows the posterior median (solid blue line) and $95\%$ Bayesian credibile interval or BCI (shaded area) of the estimated evolutionary rate under the polyepoch clock model (in blue) and the strict clock model (in red). 
The posterior median under the polyepoch clock model is almost constant and perfectly aligns with the true rate (black dashed line), except towards the present, where the sampling times are concentrated. 
The estimated rate under the strict clock model is slightly overestimated.
The $95\%$ BCIs under the polyepoch clock model become wider near the root due to less information being available to estimate the rate closer to the root. 
Also, we observe significantly higher uncertainty in estimation under the polyepoch clock  model as compared to the strict clock model, which is expected, since the latter has only one parameter associated with the clock model. 
The bottom panel shows the posterior density of the tMRCA under the polyepoch clock model (in blue) and strict clock model (in red) along with the true tMRCA indicated with a black dashed line.
We observe that the posterior density of the tMRCA has concentrated around its true value.
The bottom right panel shows the posterior density of the GMRF precision. 
The posterior median is $385$ and the $95\%$ BCI is $(49, 2520)$. 
This is consistent with the fact that the estimated rate under this simulation setup is not expected to have much variation.
Hence the polyepoch clock model performs almost as well as the strict clock model in recovering the true rate and tMRCA under this simulation setup where the strict clock represents the truth.
\begin{figure}
     \centering
   	\includegraphics[width=0.8\linewidth]{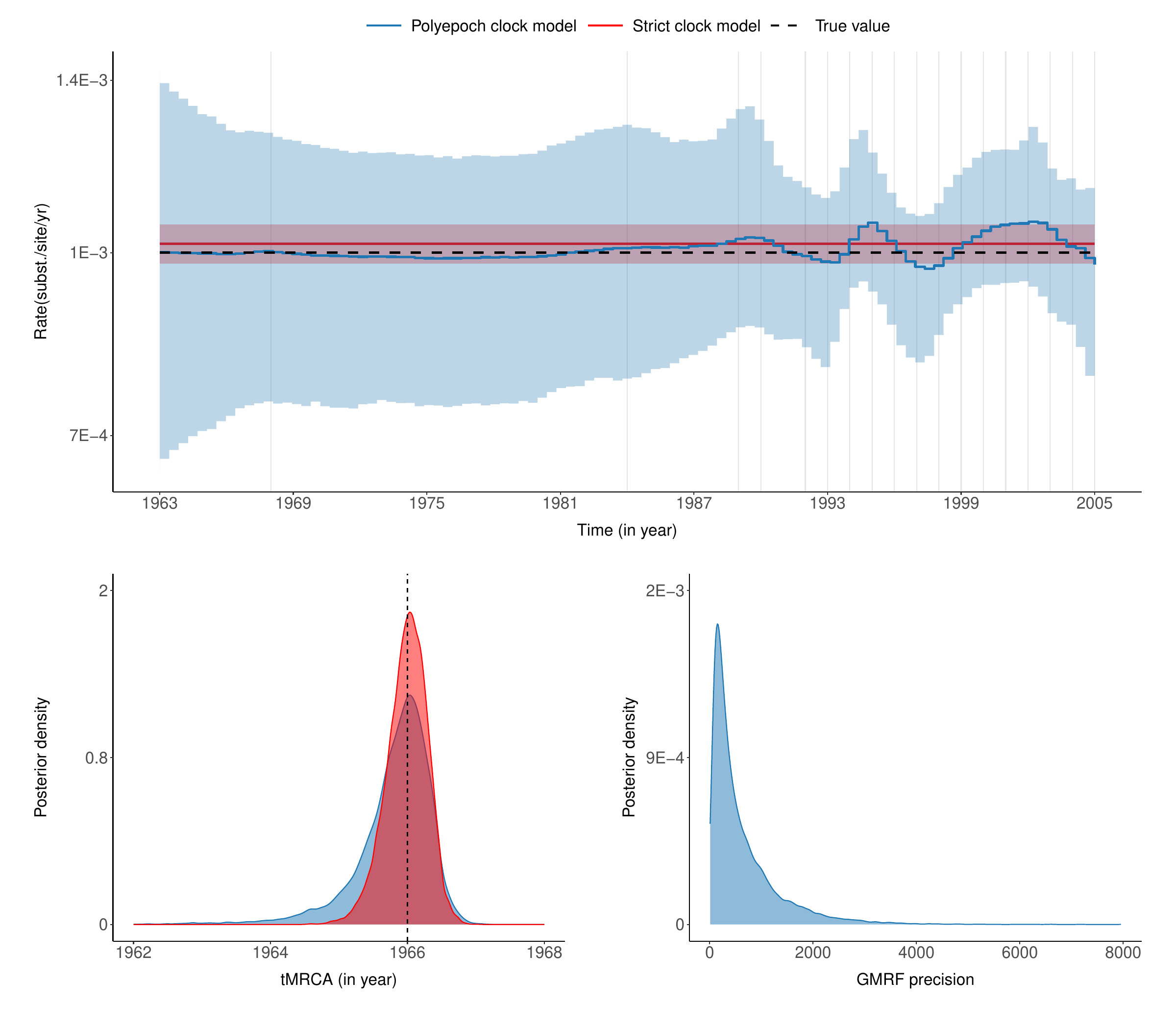} 
  	\caption{Strict clock model simulation results.
  	The top panel shows the posterior median (solid line) and the $95\%$ Bayesian credible interval or BCI (shaded area) of the evolutionary rate under the polyepoch clock model (in blue) and strict clock  model (in red). 
  	The black dashed line shows the true rate and vertical gray lines represent sampling times.
  	In the bottom left panel, the posterior density of the time of most recent common ancestor (tMRCA) under the polyepoch clock model and strict clock model is shown along with the true tMRCA (black dashed line).
  	In the bottom right panel, the posterior density of the GMRF precision under the polyepoch clock model is shown.}
  	\label{fig:scrsim}
\end{figure}
Next, we simulate sequences under a log-linear clock model with rate $f(t) = \text{e}^{-4.5 - 0.05 t}$ and using the same MCC tree, substitution process and demographic model as the previous simulation. 
We then draw inference under the polyepoch clock model with the same grid-points used in the previous analysis and also under the true log-linear rate model. 
The top panel of Figure~\ref{fig:loglinsim} shows that the posterior median under the polyepoch clock model (solid blue line) aligns more closely with the true rate (black dashed line) near the present and deviates from the true rate as we move closer to the root.
The rate under the polyepoch clock model is significantly overestimated near the root.
The estimated rate under the log-linear clock model (solid red line) perfectly aligns with the true rate and has significantly less uncertainty as compared to the polyepoch clock model due to the difference in the number of parameters to be estimated.
The bottom panel of Figure~\ref{fig:loglinsim} shows the posterior density of the tMRCA under the polyepoch clock model (in blue) and the log-linear clock model (in red),  along with the true tMRCA (black dashed line). 
The posterior density under the log-linear clock model is concentrated around the true value.
However, the tMRCA is overestimated under the polyepoch clock model, which is expected, due to the overestimation of the rate near the root.
The bottom right panel shows the posterior density of the GMRF precision under the polyepoch clock  model. 
The posterior median is $20$ and the $95\%$ BCI is $(9, 40)$. 
The relatively smaller estimated value of the GMRF precision for this simulated dataset is consistent with the expected larger variation in the estimated evolutionary rate.

\begin{figure}
     \centering
   	\includegraphics[width=0.8\linewidth]{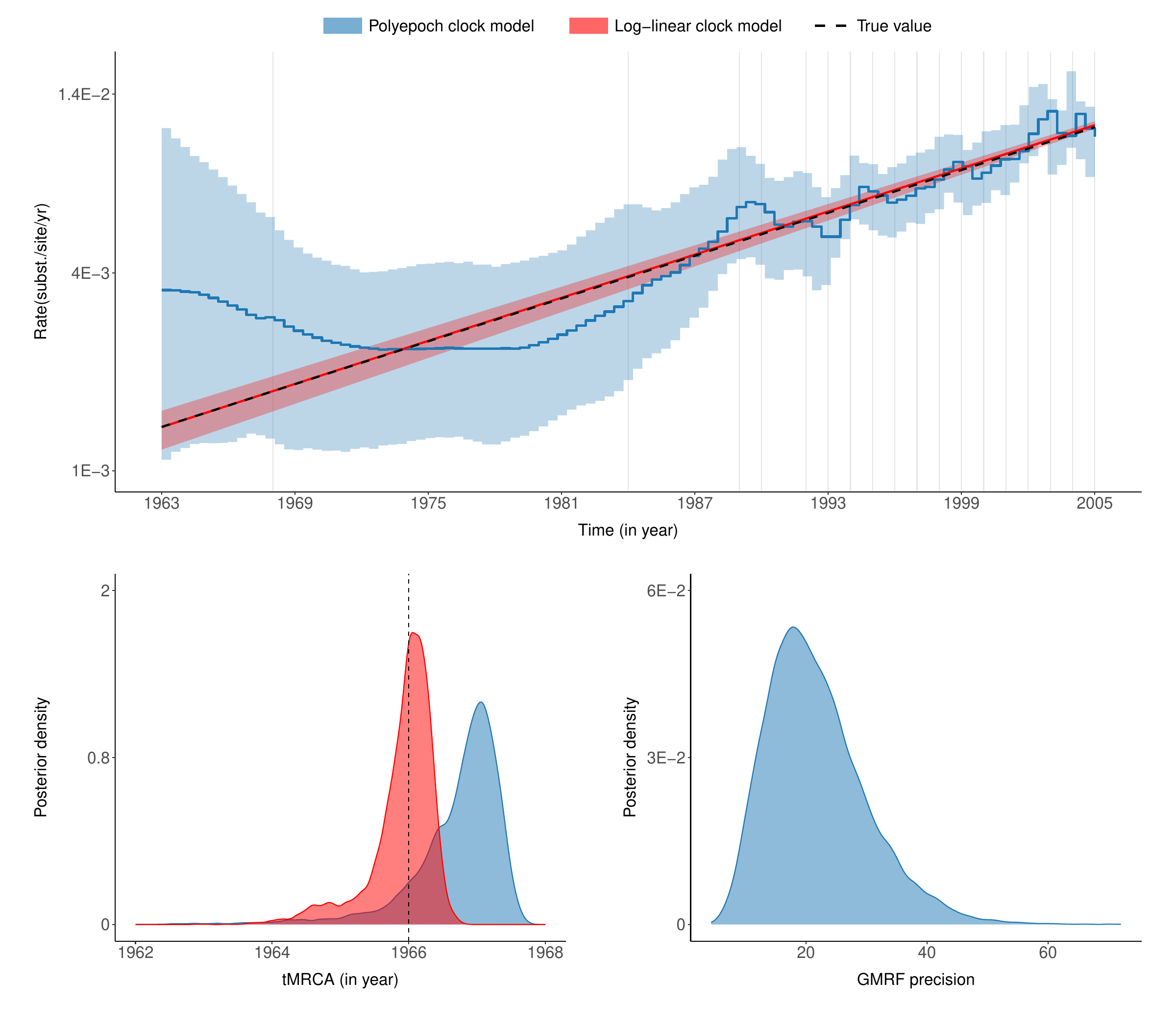} 
  	\caption{Log-linear clock model simulation results.
  	The top panel shows the posterior median (solid line) and the $95\%$ BCI (shaded area) of the evolutionary rate under the polyepoch clock model (in blue) and the log-linear clock model (in red). 
  	The black dashed line shows the true rate and vertical gray lines represent sampling times.
 	In the bottom left panel, the posterior density of the tMRCA under the polyepoch clock  model and log-linear clock model is shown along with the true tMRCA (black dashed line).
   	In the bottom right panel, the posterior density of the GMRF precision under the polyepoch clock model is shown.}
  	\label{fig:loglinsim}
\end{figure}
   
 \subsection{Exemplars}\label{subsec:rda}
     
 \noindent We perform Bayesian phylogenetic analysis under the polyepoch clock model to examine the evolutionary rate of West Nile virus in North America (1999 - 2007) \citep{Pybus2012},  serotype 3 of Dengue virus  in Brazil (1972 - 2010), seasonal influenza A/H3N2 across the world, and the spatial diffusion rate of  SARS-CoV2 in Europe between January to October 2020 \citep{Lemey2021}. 
 For the aforementioned data sets, we choose the grid-points in the polyepoch clock model to be equally spaced. 
 As a further example for model flexibility, we also apply the polyepoch clock model to examine the evolutionary rate of the Lassa virus \citep{Andersen2015} in West Africa (see Appendix~\ref{sec:Lassa_appendix}), where we choose denser grid-points during the timeframe when all sampling times are concentrated to capture a finer scale variation in the rate and more widely spaced grid-points beyond the informative timeframe.
 We use the BEAST X software package \citep{Baele2025} in combination with the high-performance phylogenetic library BEAGLE \citep{Ayres2019} to infer the evolutionary rate of the four viruses under our proposed polyepoch clock   model. 
 For the first three datasets, we assume a general-time-reversible substitution process \citep{Tavare1986} and a non-parametric tree prior  \citep{Skygrid}, where the cutoff for the tree-prior variation and the polyepoch clock model are chosen to be sufficiently greater than the root height of the unobserved coalescent process based on previous studies, to ensure all the information about the population dynamics and the evolutionary rate available from the data is captured. 
 For SARS-CoV2, we perform a joint sequence and discrete trait inference (see section~\ref{sec:sc2} for details). 
 The BEAST analyses comprise 40 million MCMC iterations for the West Nile virus data set, 20 million iterations for the Dengue virus data set, 3 million iterations for the seasonal influenza data set and 25 million iterations for the SARS-CoV2 data set, with sufficiently high effective sample size (ESS) values for all parameters of interest, as assessed using Tracer \citep{Tracer}. 
 We summarize the resulting MCC phylogenies using TreeAnnotator and plot the phynogenies using \texttt{ggtree} \citep{ggtree} in R. 
 We provide instructions and the BEAST XML files for reproducing these analyses on Github at \url{https://github.com/suchard-group/polyepoch_model_supplement.git}.
 
 \subsubsection{West Nile Virus}
 
West Nile virus (WNV) is a mosquito-borne RNA virus whose primary host is birds. 
It was first detected in the United States in New York City in August 1999, and by 2004 the infection had spread to the country's West Coast \citep{Pybus2012}.  
This example consists of 104 full genomes, with a total alignment length of 11,029 nucleotides, which were collected from infected human plasma samples from 1999 to 2007 as well as near-complete genomes obtained from GenBank \citep{Pybus2012}. 
We fit the polyepoch clock model to this dataset with 60 epochs of equal width between 1995 to 2007. 
The bottom panel of Figure~\ref{fig:wnv} shows the posterior median estimate of the evolutionary rate (solid blue line) under the polyepoch clock model, which appears to be roughly constant.
 We also fit the strict clock model to these data to compare its fit relative to the polyepoch clock model.
The posterior median evolutionary rate under the strict clock model and the polyepoch clock model are quite close as shown in Figure~\ref{fig:wnv}.
The estimated rate under the strict clock model has significantly less uncertainty compared to the polyepoch clock model.
The top panel of Figure~\ref{fig:wnv} shows the inferred MCC tree under the polyepoch clock model.
The posterior median of the tMRCA is $1998$, and its $95\%$ BCI is $\left[1997, 1999\right]$. 
We further estimate the log marginal likelihood of the data under the polyepoch clock model and strict clock model using generalized stepping stone sampling \citep{Baele2016}. 
The estimated log marginal likelihood is approximately $-25482$ under the polyepoch clock model and $-25455$ under the strict clock model. 
Hence, the strict clock model fits these data relatively better as compared to the polyepoch clock model.

\begin{figure}
     \centering
   	\includegraphics[width=0.75\linewidth]{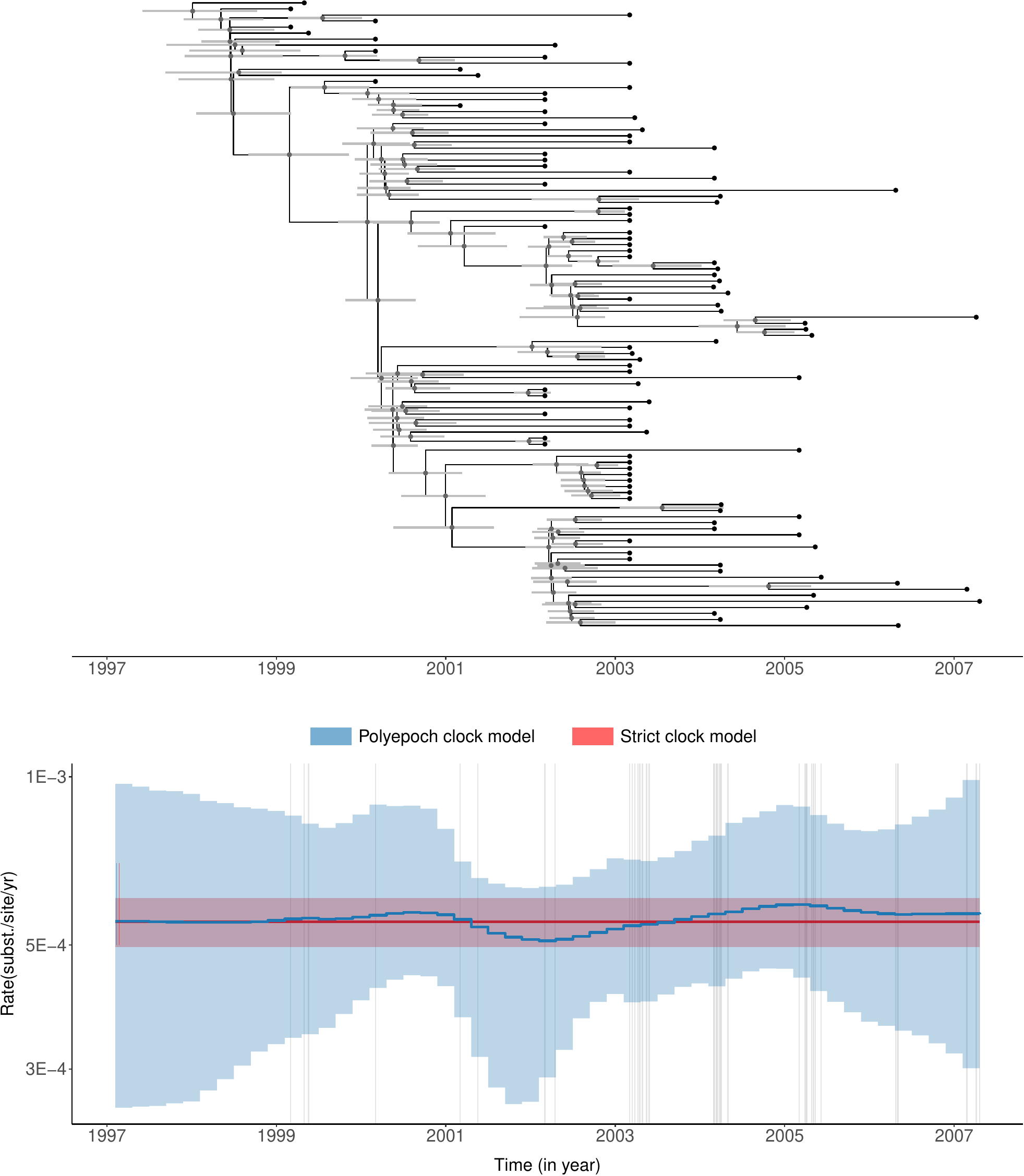} 
  	\caption{Evolutionary dynamics of West Nile virus in North America from 1997 to 2007. 
  	The top panel shows the MCC tree under the polyepoch clock model, where the posterior median and $95\%$ BCI of the internal node dates are depicted with gray circles and gray shaded bars respectively. 
  	The bottom panel shows the  posterior median (solid line) and the $95\%$ BCI (shaded area) of the evolutionary rate under the polyepoch clock model (in blue)  and the strict clock model (in red). 
  	The vertical gray lines represent sampling times.}
  	\label{fig:wnv}
\end{figure}

\subsubsection{Dengue Virus}
  Dengue fever caused by the Dengue virus is one of the most widespread mosquito borne diseases, causing around 400 million human infections each year. 
  \cite{Nunes2014} select 352 serotype 3 Dengue (DENV-3) sequences from a total of 639 complete genomes sampled between 1972 and 2010 from 31 distinct countries in Southeast Asia, North America, Central America, the Caribbean and South American countries. 
  We analyze these DENV-3 sequences under the polyepoch clock model with 200 epochs of equal width between 1970 and 2010. 
  
Figure~\ref{fig:dnv} shows the posterior median (solid blue line) and $95\%$ BCI (blue shaded area) of the evolutionary rate.
  We observe a strong time-varying pattern, particularly where the data are informative. 
  The rate decreases by a factor of almost 10 from 1995 to 2000, whereas from 2003 to 2009, the rate becomes 10 times higher. 
  The general trend appears to be decreasing from 1990 to 2000 and increasing from 2003 to 2008. 
  The most variation in the estimated rate is observed between 1990 to 2010, which coincides with the emergence of major regional lineages in the Caribbean, Venezuela, Central America and Brazil. 
  In order to ascertain whether the observed temporal variation is due to lineage-specific differences in evolutionary rate, we fit a random local clock model \citep{Drummond2010} to these data. 
  The random local clock model allows posterior identification of a series of local molecular clocks, each extending over a subregion of the full phylogeny. 
  We estimate the log marginal likelihood of the data under both clock models using generalized stepping stone sampling \citep{Baele2016}. 
  The estimated  log marginal likelihood is approximately $-57100$ under the random local clock model and $-56880$ under the polyepoch clock model.
  Hence, the polyepoch clock model yields a better fit than the random local clock model, thereby suggesting that the observed variation in the evolutionary rate of the Dengue virus cannot be entirely explained by lineage-specific rate shifts.
  By allowing flexible rate changes over time, the polyepoch clock model provides a more accurate representation of the temporal variation in the evolutionary rate of these Dengue viruses.
  The top panel of Figure~\ref{fig:dnv} shows the inferred MCC tree under the polyepoch clock model, where tip nodes are colored according to the geographical location and the posterior median and $95\%$ BCI of the internal node dates are indicated with grey circles and gray shaded bars respectively.
  We identify the two Brazilian lineages I and II as found in the original study \cite{Nunes2014}. 
  The posterior median of the tMRCA is $1972$ (posterior median) and its $95\%$ BCI is $(1963, 1973)$. 

\begin{figure}
     \centering
   	\includegraphics[width=0.75\linewidth]{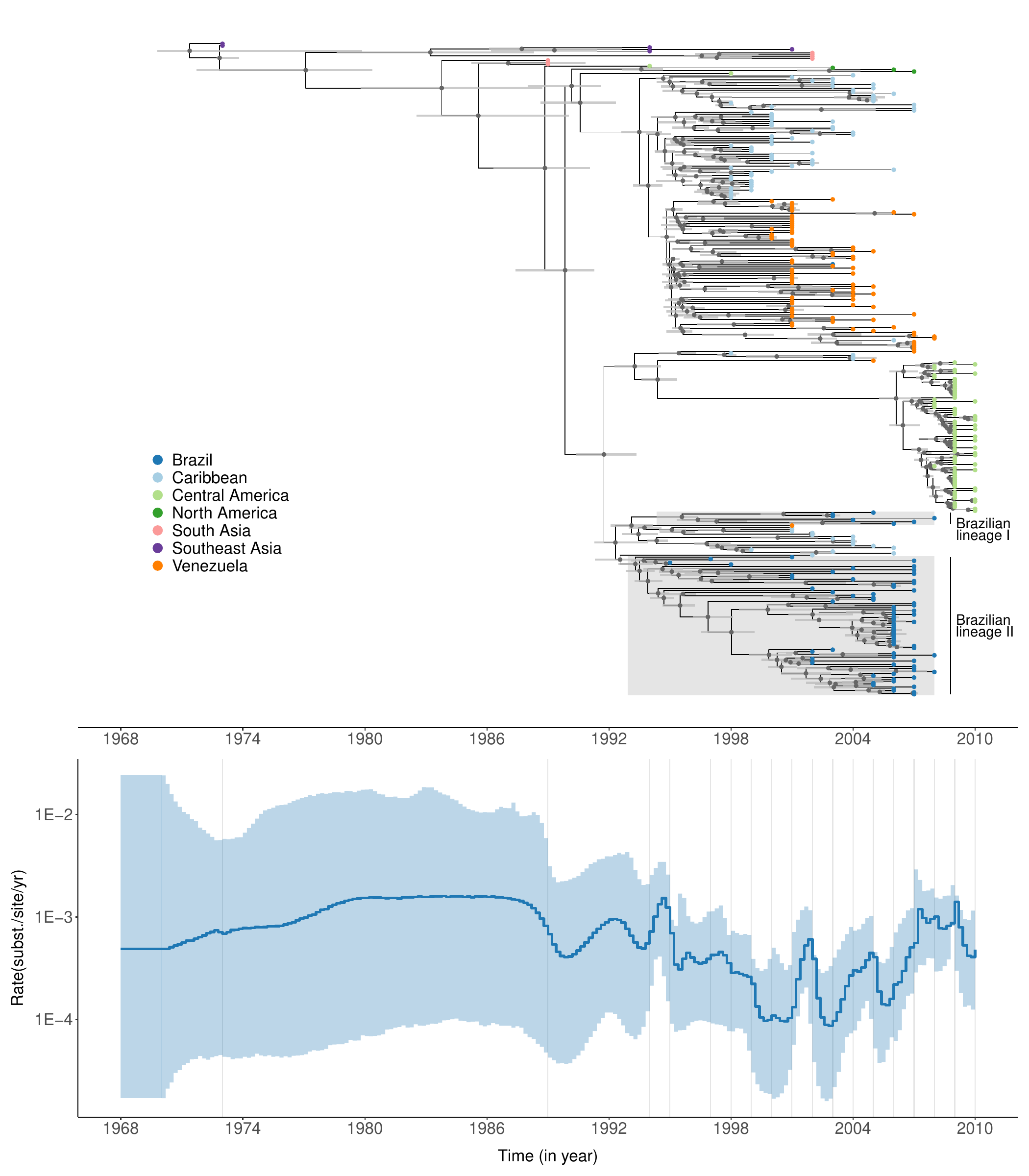} 
  	\caption{Evolutionary dynamics of Dengue virus from 1968 to 2010 under the polyepoch clock model.
  	The top panel shows the MCC tree. 
  	  	The tips are colored according to geographical location and the posterior median and $95\%$ BCI of internal node dates are depicted with gray circles and gray shaded bars respectively.
  	  	The two Brazilian lineages I and II found in the original study are highlighted in gray.
  	The bottom panel shows the posterior median (solid blue line) and the $95\%$ BCI (blue shaded area) of the evolutionary rate from 1970 to 2010. 
  	Vertical gray lines represent sampling times. }
  	\label{fig:dnv}
\end{figure}  

 \subsubsection{Seasonal Influenza A/H3N2}
 Seasonal influenza
 is an acute respiratory infection caused by influenza viruses. 
 There are around a billion cases of seasonal influenza annually, including 290,000 to 650,000 respiratory deaths annually. 
 Among the different types of influenza viruses, influenza A is highly virulent and causes seasonal epidemics in humans, mainly during winter in non-tropical regions. 
 We study the temporal variation in the evolutionary rate of the H3N2 subtype of influenza A (A/H3N2). 
 The data set contains 402 A/H3N2 sequences sampled globally from 1968 to 2010.  
 We fit the polyepoch clock model with 360 epochs of equal width from 1966 to 2010. 
 Much like the seasonal pattern in influenza incidence, there may also be seasonal variation in the evolutionary rate. 
 To avoid oversmoothing that could obscure such potential cyclical patterns, we assign a logit-normal prior to $\rho$ with location parameter 0 and scale parameter 1. 
 
  The top panel of Figure~\ref{fig:h3n2} shows the inferred MCC tree under the polyepoch clock model, where tip nodes are colored according to the geographical location and the posterior median and $95\%$ BCI of the internal node dates are indicated with grey circles and gray shaded bars respectively.
 The bottom panel of Figure~\ref{fig:h3n2} shows the posterior median and $95\%$ BCI of the evolutionary rate under the polyepoch clock model. 
 We observe a seasonal pattern with peaks occurring between December and February each year, which is the mid-point of the influenza season in the Northern hemisphere. 
 We also observe that the peak becomes higher during the recent years, from 2001 to 2011. 
 The posterior median of the tMRCA is 1965 and the $95\%$ BCI is $(1961, 1966)$. 
 The posterior median of $\rho$ is $0.26$ and the $95\%$ BCI is $(0.07, 0.58)$. We note that the posterior density of $\rho$ has no mass approaching 1. This highlights the importance of
  estimating $\rho$ for this example to avoid over-smoothing for rapidly changing rates.  
 The prior and posterior density of $\rho$ are significantly different (see density plot in top panel of Figure~\ref{fig:h3n2}), which suggests that our model remains robust to the choice of the hyperprior for $\rho$.
 
\begin{figure}
     \centering
   	\includegraphics[width=0.8\linewidth]{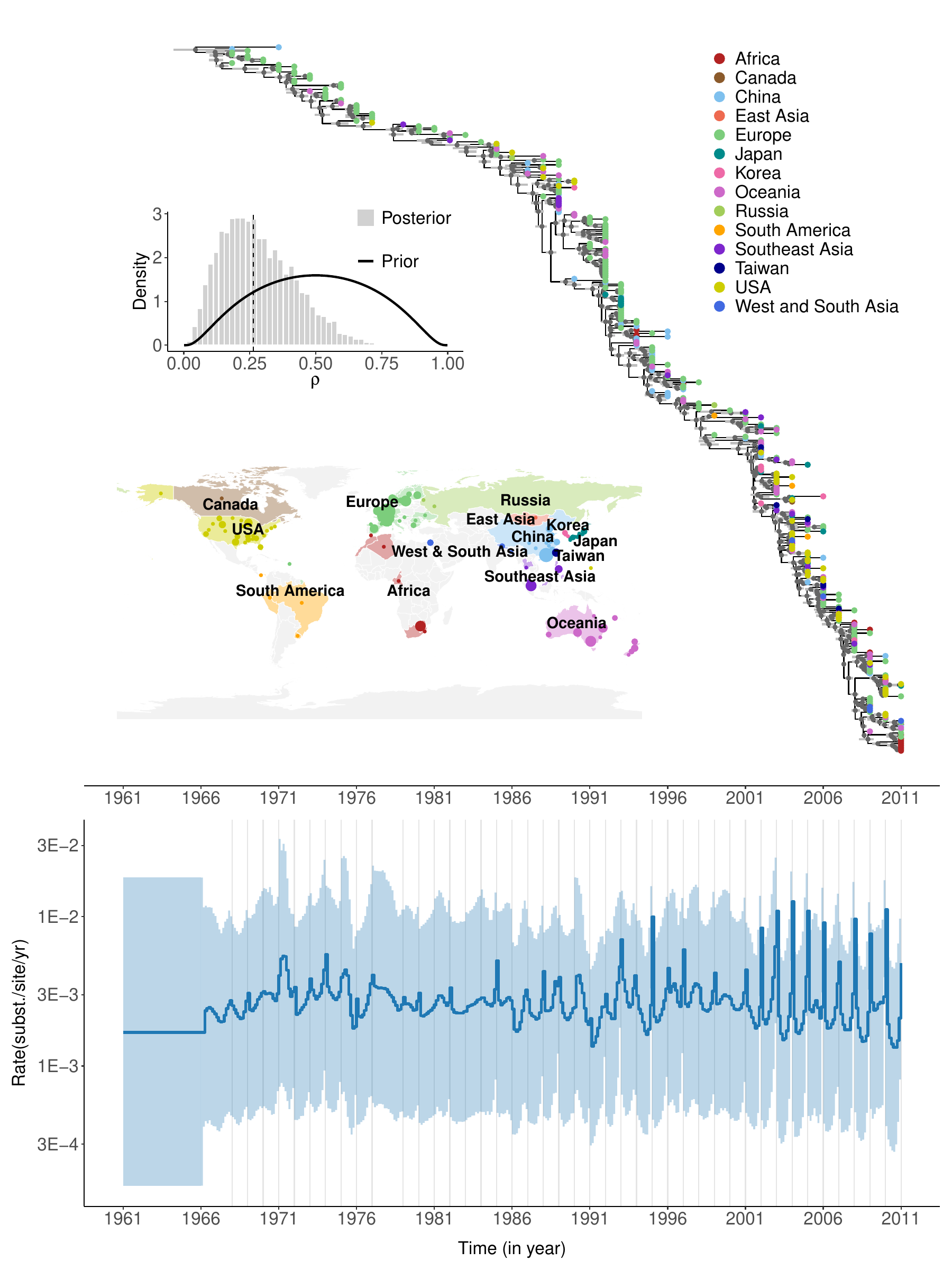} 
  	\caption{Evolutionary dynamics of seasonal influenza A/H3N2 across the globe from 1961 to 2011 under the polyepoch clock model.
  	The top panel shows the MCC tree, where the tips are colored according to geographic location and the posterior median and $95\%$ BCI of internal node dates are depicted with gray circles and gray shaded bars respectively. 
  	The geographic areas for which A/H3N2 sequence samples were used in the analysis are shown with colored dots on the map. 
  	The gray density plot shows the posterior density of $\rho$ and the prior density is shown with a solid black line. 
  	The bottom panel shows the posterior median (solid blue line) and the $95\%$ BCI (blue shaded area) of the evolutionary rate from 1961 to 2011. 
  	Vertical gray lines represent sampling times.}
  	\label{fig:h3n2}
\end{figure}  
 
  \subsubsection{SARS-CoV-2 Resurgence in Europe}\label{sec:sc2}
  
  SARS-CoV-2 is responsible for the COVID-19 pandemic with an estimated 7 million deaths worldwide. 
  After the first wave of SARS-CoV-2 infections in Spring 2020, the number of COVID-19 infections started rising rapidly in many European countries during late Summer 2020, gradually leading to a second wave by October 2020. 
  We analyze 3959 SARS-CoV2 genomes from Belgium, France, Germany, Italy, the Netherlands, Norway, Portugal, Spain, Switzerland and the UK, from both the first and second wave, collected from GISAID on November 3, 2020 \citep{Lemey2021}.
  We perform a joint Bayesian sequence and trait inference by incorporating the genome sequences, the country and date of sampling, as well as mobility and connectivity data. 
  We model the effective population size trajectory assuming a semi-parametric 
 piecewise-constant function, 
where the log-transformed population sizes are modelled as a deterministic function of log-transformed counts of cases of COVID-19 over 2 week intervals \citep{Skygrid}.
  We model the process of transitioning through discrete location states, which are the countries of sampling in this case, using the polyepoch clock model to capture overall rate variation through time, while assuming the relative rates are log-linear functions of mobility and connectivity covariates following the original study \citep{Lemey2021}.
  Each epoch corresponds to a 2-week interval, same as the epoch specification for the tree prior, and covariates include the social connectedness index (SCI) of Facebook, air transportation data and mobility data. 
  We estimate the inclusion probability of each potential covariate through a spike-and-slab procedure \citep{Lemey2014}. 
  
 The top panel of Figure~\ref{fig:sc2_pcm} shows the posterior median (solid blue line) and $95\%$ BCI (blue shaded area) of the rate of spatial diffusion under the polyepoch clock   model. 
 The rate decreases by almost a factor of one tenth from mid December 2019 to late March 2020. 
 This sharp drop aligns with the lockdown in Spring 2020 in response to the first wave. 
 From April to mid-July 2020, there is a gradual increase in the rate by almost a factor of 9. This aligns with the time period following the containment of the first wave, leading to partial relaxation of restrictions during early summer 2020. 
 With the onset of the second wave, the number of cases began rising in late summer, which subsequently led to another lockdown. 
 We observe a gradual drop in the rate during this second lockdown by almost a factor of one-ninth from mid-July 2020 to October 30, 2020.
 
 The middle panel of Figure~\ref{fig:sc2_pcm} shows the posterior median (solid orange line) and $95\%$ BCI (orange shaded area) of the effective population size as obtained under the skygrid \cite{Skygrid}. 
 We observe a sharp increase in the effective population size, by almost a factor of 250, from December 2019 to early April 2020.
 This is followed by a gradual decrease from early April 2020 to early July 2020, by almost a factor of 2.5, and then an increase in from mid-July 2020 to October 2020, by a factor of about 4. 
 We observe from Figure~\ref{fig:sc2_pcm} that the spatial diffusion rate and the effective population size trajectories exhibit negatively correlated trends.

\begin{figure}
        \centering
        \includegraphics[width=0.75\linewidth]{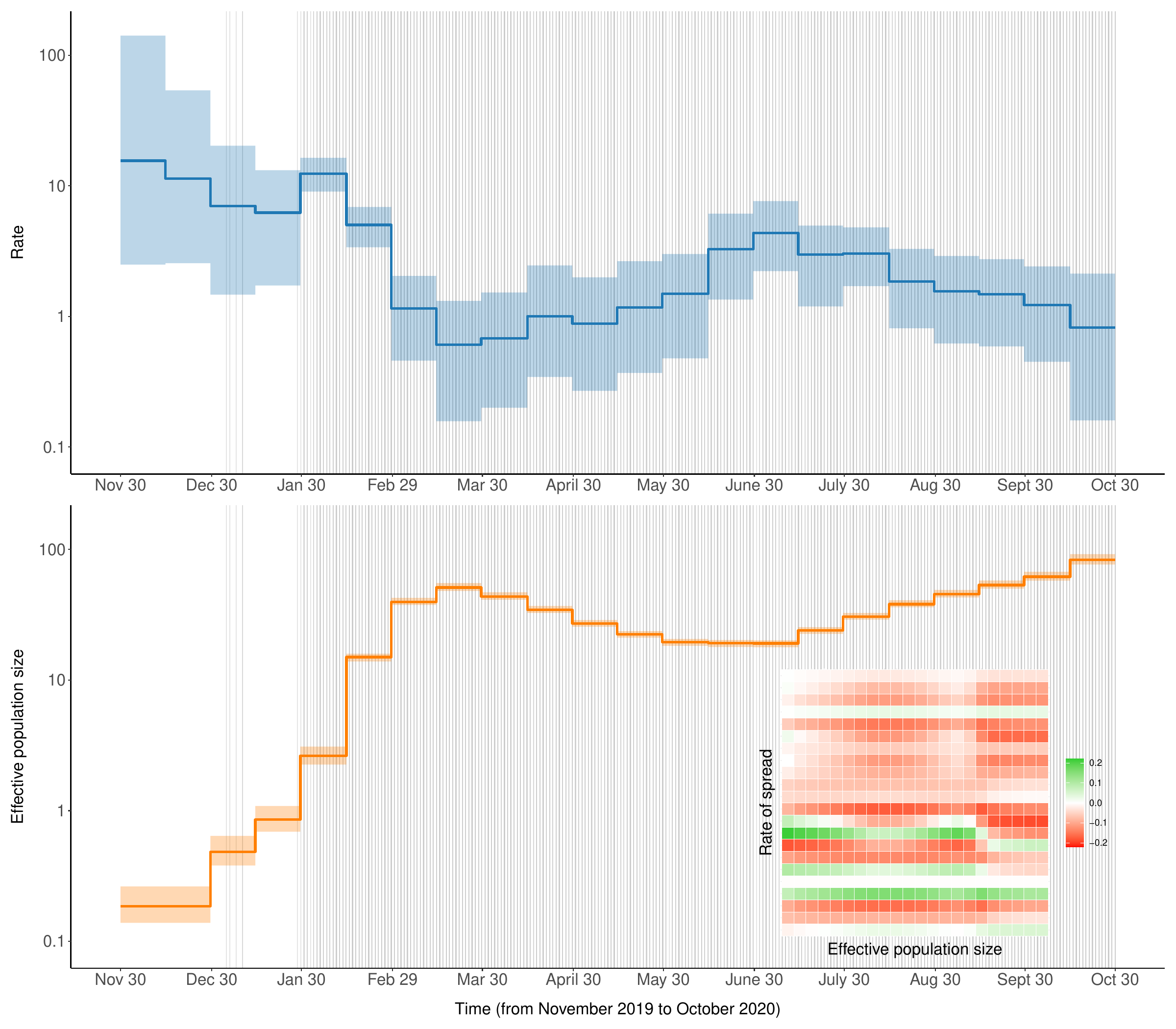} 
    	\caption{Evolutionary dynamics of SARS-CoV-2 in Europe from November 2019 to October 2020 under the polyepoch clock model.
    	The top panel shows the posterior median (solid blue line) and the $95\%$ BCI (blue shaded area) of the rate of spread of SARS-CoV-2 between 10 European countries under the polyepoch clock model.
    	The bottom panel shows the posterior median (solid orange line) and the $95\%$ BCI (orange shaded area) of the effective population size under the skygrid model.
    	The correlation matrix on the bottom right shows the correlation between the rate of spread (along rows) and the effective population size (along columns).
    	Green indicates positive correlation and red indicates negative correlation.
    	Darker shade indicates higher absolute value.
    	Vertical gray lines indicate sampling times.}
    	\label{fig:sc2_pcm}
\end{figure}  

The bottom right panel of Figure~\ref{fig:sc2_pcm} shows the correlation matrix between the posterior samples of the rate of spread under the polyepoch clock model along the rows and the effective population size along the columns.
We indicate positive correlation with green and negative correlation with red, and higher absolute values are indicated with a darker shade. 
We observe that most of the rows are red, suggesting that there is a negative correlation between the rate and the effective population size.
This negative correlation is plausibly explained by multiple introductions of the virus prior to extensive local transmission early in the pandemic.
Once local transmission dominates, there is an increase in the effective population size and  introductions do not contribute meaningfully at this stage.
The absolute value of the correlation decreases as the distance between the epochs increases.

 The inferred MCC tree for this analysis is shown in Figure~\ref{fig:sc2tree}. 
 We observe that the MRCA originated in December 2019. 
 Our analysis captures the spread of the variants B.1.160/20A.EU2 and B.1.177/20E(EU1) (see highlighted clades in Figure~\ref{fig:sc2tree}) in Europe during 2020. Our findings are consistent with \cite{Lemey2021} in identifying the variant B.1.177/20E(EU1) to have originated in Spain, which is then transmitted onwards to rapidly become the most dominant strain in the UK during late summer 2020. 
 We also identify the variant B.1.160/20A.EU which largely spread out of France during summer 2020. 

\begin{figure}
            \centering
            \includegraphics[width=0.75\linewidth]{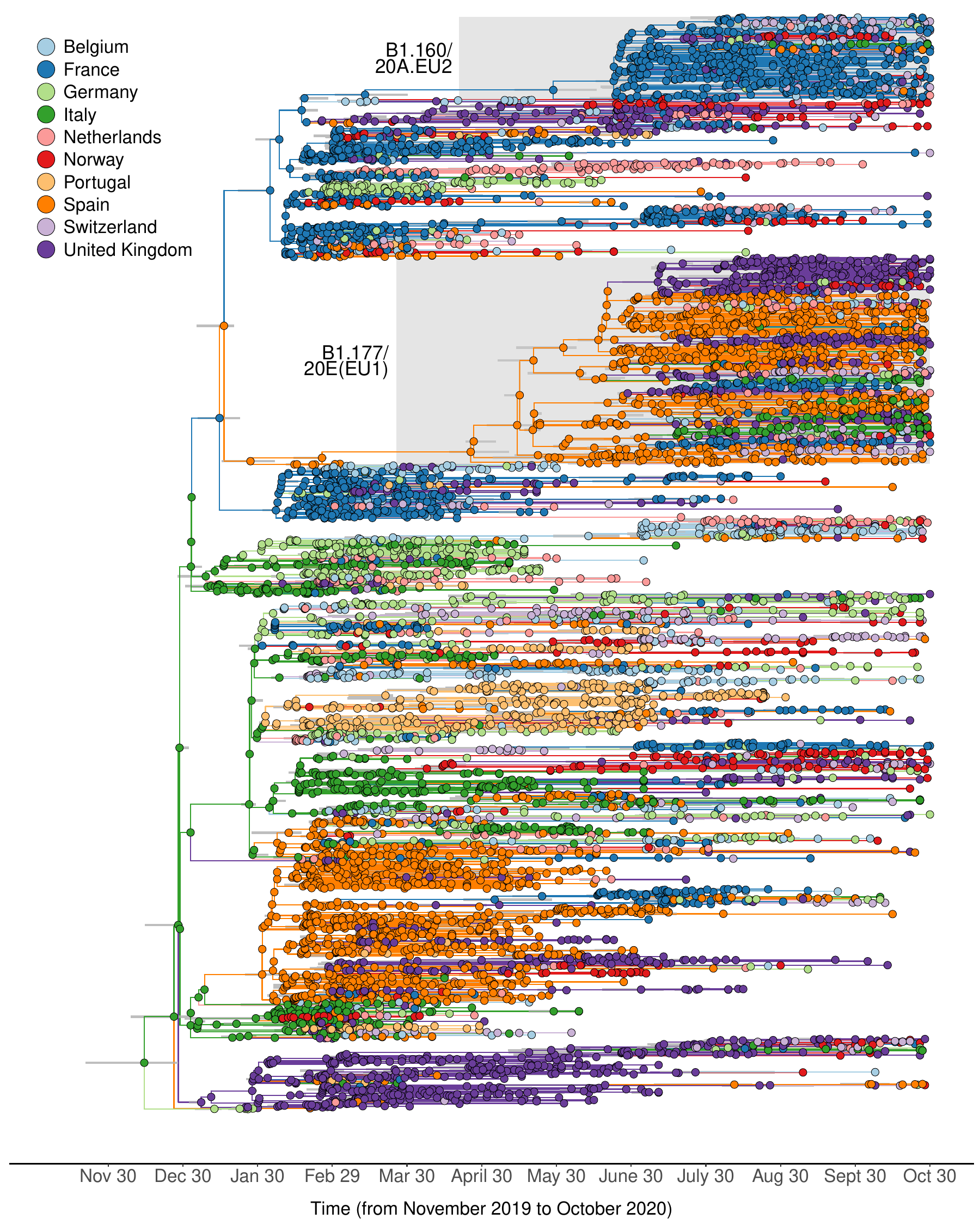} 
        	\caption{MCC tree obtained under the polyepoch clock model for the spread of SARS-CoV-2 in Europe from November 2019 to October 2020. 
        	Nodes and branches are colored according to the countries.}
        	\label{fig:sc2tree}
\end{figure}
  
 \section{Discussion}
 \label{sec:Discussion}
 
 Building on recent insights into the time-varying nature of evolutionary rates, particularly for rapidly evolving viruses \citep{Aiewsakun2016, Duchene2014, aiewsakun2015b, Membrebe2019}, we develop the polyepoch clock   model that is a novel and flexible approach to estimate the change in evolutionary rate through time.
 The polyepoch clock model accommodates the temporally varying nature of the evolutionary rate by modeling molecular sequences to arise from ICTMCs. 
 We model the evolutionary rate as an unknown, positive and integrable function, which leads to closed-form expressions of the transition probability matrix of the ICTMC \citep{Fortmann1977, Kailath1980, Rindos}, expressed in terms of the integral of the function describing the evolutionary rate through time, which needs to be evaluated for all the branches of the phylogeny. 
 We propose the polyepoch clock model as a solution to this problem by modeling the evolutionary rate as a piecewise constant function with a large number of epochs. 
 This retains the flexibility of the estimated rate and makes the computation of the transition probabilities tractable and inexpensive. 
 We use a class of proper GMRF \citep{rueheld} priors to incorporate the temporal dependence between the rate in adjacent epochs, while ensuring that the posterior remains proper.
Given the high dimensionality and strong correlations among rate parameters, we leverage the gradient-based HMC sampler \citep{Neal2011} to efficiently sample from the posterior distribution of the rate parameters. 
To further enhance computational efficiency, we combine a linear-time gradient evaluation algorithm with respect to branch-specific parameters with the chain rule to derive gradients with respect to the rate parameters. 
This approach incurs substantially less computational cost as compared to standard alternatives such as the peeling algorithm \citep{Felsenstein1973, Felsenstein1981} or numerical approximations based on finite differences, making our inference framework scalable to larger datasets.
 
 We assess the performance of the polyepoch clock model in recovering the true evolutionary rate and tMRCA under two different simulation scenarios. 
 Because the rate and tree are estimated jointly, uncertainty in the rate naturally increases near the root, particularly when the data lack calibration points or are sparsely sampled over time. 
 Nevertheless, when the data are informative, the polyepoch clock model accurately recovers the true rate trajectory.
 We then apply the model to four viral examples that exhibit diverse patterns of evolutionary rate variation.
 Our first example, which is the West Nile virus \citep{Pybus2012} in North America from 1999-2007, resembles a relatively constant rate.
 The second example of the Dengue virus \citep{Nunes2014} in Brazil from 1972 - 2010 exhibits significant variation in the evolutionary rate within the sampling time range. 
 We also show that the polyepoch clock model offers a relatively better fit to the Dengue virus data as compared to the random local clock model, which could be a plausible choice given that multiple predominant clades emerge during the time period when most of the variation in the inferred rate is observed under the polyepoch clock model.
 We also apply our model to study the evolutionary dynamics of seasonal influenza A/H3N2, where we assign a logit normal prior to the propriety parameter of the GMRF prior.
 This allows the model to flexibly adapt the degree of smoothing and successfully capture  seasonal peaks in the rate trajectory.
 We further analyze SARS-CoV-2 spread in Europe during January–October 2020 \citep{Lemey2021}. 
 The inferred temporal variation in the spatial diffusion rate aligns well with real-world interventions, showing marked decreases during lockdowns and increases during periods of relaxed restrictions. 
 Moreover, our analysis reveals a negative correlation between the inferred diffusion rate and the effective population size. 
 This motivates us to pursue modeling the log-transformed rate and log-transformed effective population size using bivariate conditionally autoregressive (CAR) models \citep{Gelfand2003} as a potential future direction.

  We implement the polyepoch clock model in BEAST X \citep{Baele2025} providing flexibility in using it with different choices of tree priors \citep{Skygrid}, different substitution processes \citep{Tavare1986, hky1985} and internal and tip node calibration \citep{Rambaut2000, YangRannala2005}. 
  This also allows us to assess the relative fit of the polyepoch clock model compared to the suite of clock models available in BEAST, leveraging the software’s %state-of-the-art 
  marginal likelihood estimation framework \citep{Baele2012, Baele2016}. 
  We hope that the polyepoch clock   model will enhance the accuracy in estimating the change in evolutionary rate through time for rapidly evolving viruses as well as other biological entities and offer novel insights into long-term evolutionary aspects.

\begin{acks}[Acknowledgments]
This work was supported through National Institutes of Health grants R01 AI153044 and R01 AI162611.
We gratefully acknowledge support from Advanced Micro Devices, Inc.~with the donation of parallel computing resources used for this research.
PL acknowledges support by the Research Foundation - Flanders (‘Fonds voor Wetenschappelijk Onderzoek - Vlaanderen’, G005323N and G051322N).
\end{acks}

%% if your bibliography is in bibtex format, uncomment commands:
\bibliographystyle{imsart-nameyear} % Style BST file
\bibliography{refs}

\begin{thebibliography}{43}
% BibTex style file: imsart-nameyear.bst, 2017-11-03
% Default style options (sort=1,type=nameyear).
% Used options (sort=1,type=nameyear).

\bibitem[\protect\citeauthoryear{Aiewsakun and
  Katzourakis}{2015}]{aiewsakun2015b}
\begin{barticle}[author]
\bauthor{\bsnm{Aiewsakun},~\bfnm{Pakorn}\binits{P.}} \AND
  \bauthor{\bsnm{Katzourakis},~\bfnm{Aris}\binits{A.}}
(\byear{2015}).
\btitle{Time dependency of foamy virus evolutionary rate estimates}.
\bjournal{BMC Evolutionary Biology}
\bvolume{15}
\bpages{119}.
\bdoi{10.1186/s12862-015-0408-z}
\end{barticle}
\endbibitem

\bibitem[\protect\citeauthoryear{Aiewsakun and
  Katzourakis}{2016}]{Aiewsakun2016}
\begin{barticle}[author]
\bauthor{\bsnm{Aiewsakun},~\bfnm{Pakorn}\binits{P.}} \AND
  \bauthor{\bsnm{Katzourakis},~\bfnm{Aris}\binits{A.}}
(\byear{2016}).
\btitle{Time-Dependent Rate Phenomenon in Viruses}.
\bjournal{Journal of Virology}
\bvolume{90}
\bpages{7184-7195}.
\bdoi{10.1128/jvi.00593-16}
\end{barticle}
\endbibitem

\bibitem[\protect\citeauthoryear{Andersen et~al.}{2015}]{Andersen2015}
\begin{barticle}[author]
\bauthor{\bsnm{Andersen},~\bfnm{Kristian~G.}\binits{K.~G.}},
  \bauthor{\bsnm{Shapiro},~\bfnm{B.~Jesse}\binits{B.~J.}},
  \bauthor{\bsnm{Matranga},~\bfnm{Christian~B.}\binits{C.~B.}},
  \bauthor{\bsnm{Sealfon},~\bfnm{Rachel S.~G.}\binits{R.~S.~G.}},
  \bauthor{\bsnm{Lin},~\bfnm{Aaron~E.}\binits{A.~E.}},
  \bauthor{\bsnm{Moses},~\bfnm{Lina~M.}\binits{L.~M.}},
  \bauthor{\bsnm{Folarin},~\bfnm{Onikepe~A.}\binits{O.~A.}},
  \bauthor{\bsnm{Goba},~\bfnm{Augustine}\binits{A.}},
  \bauthor{\bsnm{Odia},~\bfnm{Ikponmwonsa}\binits{I.}},
  \bauthor{\bsnm{Ehiane},~\bfnm{Philomena~E.}\binits{P.~E.}},
  \bauthor{\bsnm{Momoh},~\bfnm{Mambu}\binits{M.}},
  \bauthor{\bsnm{England},~\bfnm{Eleina~M.}\binits{E.~M.}},
  \bauthor{\bsnm{Winnicki},~\bfnm{Sarah~M.}\binits{S.~M.}},
  \bauthor{\bsnm{Branco},~\bfnm{Luis~M.}\binits{L.~M.}},
  \bauthor{\bsnm{Gire},~\bfnm{Stephen}\binits{S.}},
  \bauthor{\bsnm{Phelan},~\bfnm{Eric}\binits{E.}},
  \bauthor{\bsnm{Tariyal},~\bfnm{Ridhi}\binits{R.}},
  \bauthor{\bsnm{Tewhey},~\bfnm{Ryan}\binits{R.}},
  \bauthor{\bsnm{Omoniwa},~\bfnm{Omowunmi}\binits{O.}},
  \bauthor{\bsnm{Fullah},~\bfnm{Mohammed}\binits{M.}},
  \bauthor{\bsnm{Fonnie},~\bfnm{Richard}\binits{R.}},
  \bauthor{\bsnm{Fonnie},~\bfnm{Mbalu}\binits{M.}},
  \bauthor{\bsnm{Kanneh},~\bfnm{Lansana~D.}\binits{L.~D.}},
  \bauthor{\bsnm{Jalloh},~\bfnm{Simbirie~C}\binits{S.~C.}},
  \bauthor{\bsnm{Gbakie},~\bfnm{Michael~A}\binits{M.~A.}},
  \bauthor{\bsnm{Saffa},~\bfnm{Sidiki}\binits{S.}},
  \bauthor{\bsnm{Karbo},~\bfnm{Kandeh}\binits{K.}},
  \bauthor{\bsnm{Gladden},~\bfnm{Adrianne~D.}\binits{A.~D.}},
  \bauthor{\bsnm{Qu},~\bfnm{James}\binits{J.}},
  \bauthor{\bsnm{Stremlau},~\bfnm{Matthew~H.}\binits{M.~H.}},
  \bauthor{\bsnm{Nekoui},~\bfnm{Mahan}\binits{M.}},
  \bauthor{\bsnm{Finucane},~\bfnm{Hilary~K.}\binits{H.~K.}},
  \bauthor{\bsnm{Tabrizi},~\bfnm{Shervin}\binits{S.}},
  \bauthor{\bsnm{Vitti},~\bfnm{Joseph~J.}\binits{J.~J.}},
  \bauthor{\bsnm{Birren},~\bfnm{Bruce~W.}\binits{B.~W.}},
  \bauthor{\bsnm{Fitzgerald},~\bfnm{Michael~G.}\binits{M.~G.}},
  \bauthor{\bsnm{McCowan},~\bfnm{Caryn}\binits{C.}},
  \bauthor{\bsnm{Ireland},~\bfnm{Andrea~T.}\binits{A.~T.}},
  \bauthor{\bsnm{Berlin},~\bfnm{Aaron~M.}\binits{A.~M.}},
  \bauthor{\bsnm{Bochicchio},~\bfnm{James}\binits{J.}},
  \bauthor{\bsnm{Taz{\'o}n-Vega},~\bfnm{B{\'a}rbara}\binits{B.}},
  \bauthor{\bsnm{Lennon},~\bfnm{Niall~J.}\binits{N.~J.}},
  \bauthor{\bsnm{Ryan},~\bfnm{Elizabeth~M.}\binits{E.~M.}},
  \bauthor{\bsnm{Bjornson},~\bfnm{Zach}\binits{Z.}},
  \bauthor{\bsnm{Milner},~\bfnm{Danny~Arnold}\binits{D.~A.}},
  \bauthor{\bsnm{Lukens},~\bfnm{Amanda~K.}\binits{A.~K.}},
  \bauthor{\bsnm{Broodie},~\bfnm{Nisha}\binits{N.}},
  \bauthor{\bsnm{Rowland},~\bfnm{Megan~M.}\binits{M.~M.}},
  \bauthor{\bsnm{Heinrich},~\bfnm{Megan~L.}\binits{M.~L.}},
  \bauthor{\bsnm{Akda{\u g}},~\bfnm{Marjan}\binits{M.}},
  \bauthor{\bsnm{Schieffelin},~\bfnm{John~S.}\binits{J.~S.}},
  \bauthor{\bsnm{Levy},~\bfnm{Danielle~C.}\binits{D.~C.}},
  \bauthor{\bsnm{Akpan},~\bfnm{Henry~Dan}\binits{H.~D.}},
  \bauthor{\bsnm{Bausch},~\bfnm{Daniel~G}\binits{D.~G.}},
  \bauthor{\bsnm{Rubins},~\bfnm{Kathleen~H.}\binits{K.~H.}},
  \bauthor{\bsnm{McCormick},~\bfnm{Joseph~B.}\binits{J.~B.}},
  \bauthor{\bsnm{Lander},~\bfnm{Eric~S.}\binits{E.~S.}},
  \bauthor{\bsnm{G{\"u}nther},~\bfnm{Stephan}\binits{S.}},
  \bauthor{\bsnm{Hensley},~\bfnm{Lisa~E}\binits{L.~E.}},
  \bauthor{\bsnm{Okogbenin},~\bfnm{Sylvanus~Akhalufo}\binits{S.~A.}},
  \bauthor{\bsnm{Schaffner},~\bfnm{Stephen~F.}\binits{S.~F.}},
  \bauthor{\bsnm{Okokhere},~\bfnm{Peter~O}\binits{P.~O.}},
  \bauthor{\bsnm{Khan},~\bfnm{S.~Humarr}\binits{S.~H.}},
  \bauthor{\bsnm{Grant},~\bfnm{Donald~S.}\binits{D.~S.}},
  \bauthor{\bsnm{Akpede},~\bfnm{George~O}\binits{G.~O.}},
  \bauthor{\bsnm{Asogun},~\bfnm{Danny}\binits{D.}},
  \bauthor{\bsnm{Gnirke},~\bfnm{Andreas}\binits{A.}},
  \bauthor{\bsnm{Levin},~\bfnm{Joshua~Z.}\binits{J.~Z.}},
  \bauthor{\bsnm{Happi},~\bfnm{Christian~T.}\binits{C.~T.}},
  \bauthor{\bsnm{Garry},~\bfnm{Robert~F.}\binits{R.~F.}} \AND
  \bauthor{\bsnm{Sabeti},~\bfnm{Pardis~C}\binits{P.~C.}}
(\byear{2015}).
\btitle{Clinical Sequencing Uncovers Origins and Evolution of Lassa Virus}.
\bjournal{Cell}
\bvolume{162}
\bpages{738-750}.
\bdoi{10.1016/j.cell.2015.07.020}
\end{barticle}
\endbibitem

\bibitem[\protect\citeauthoryear{Aris-Brosou and Yang}{2003}]{Yang2003}
\begin{barticle}[author]
\bauthor{\bsnm{Aris-Brosou},~\bfnm{St{\'e}phane}\binits{S.}} \AND
  \bauthor{\bsnm{Yang},~\bfnm{Ziheng}\binits{Z.}}
(\byear{2003}).
\btitle{{Bayesian Models of Episodic Evolution Support a Late Precambrian
  Explosive Diversification of the Metazoa}}.
\bjournal{Molecular Biology and Evolution}
\bvolume{20}
\bpages{1947-1954}.
\bdoi{10.1093/molbev/msg226}
\end{barticle}
\endbibitem

\bibitem[\protect\citeauthoryear{Ayres et~al.}{2019}]{Ayres2019}
\begin{barticle}[author]
\bauthor{\bsnm{Ayres},~\bfnm{Daniel~L}\binits{D.~L.}},
  \bauthor{\bsnm{Cummings},~\bfnm{Michael~P}\binits{M.~P.}},
  \bauthor{\bsnm{Baele},~\bfnm{Guy}\binits{G.}},
  \bauthor{\bsnm{Darling},~\bfnm{Aaron~E}\binits{A.~E.}},
  \bauthor{\bsnm{Lewis},~\bfnm{Paul~O}\binits{P.~O.}},
  \bauthor{\bsnm{Swofford},~\bfnm{David~L}\binits{D.~L.}},
  \bauthor{\bsnm{Huelsenbeck},~\bfnm{John~P}\binits{J.~P.}},
  \bauthor{\bsnm{Lemey},~\bfnm{Philippe}\binits{P.}},
  \bauthor{\bsnm{Rambaut},~\bfnm{Andrew}\binits{A.}} \AND
  \bauthor{\bsnm{Suchard},~\bfnm{Marc~A}\binits{M.~A.}}
(\byear{2019}).
\btitle{BEAGLE 3: Improved Performance, Scaling, and Usability for a
  High-Performance Computing Library for Statistical Phylogenetics}.
\bjournal{Systematic Biology}
\bvolume{68}
\bpages{1052-1061}.
\bdoi{10.1093/sysbio/syz020}
\end{barticle}
\endbibitem

\bibitem[\protect\citeauthoryear{Baele, Lemey and Suchard}{2016}]{Baele2016}
\begin{barticle}[author]
\bauthor{\bsnm{Baele},~\bfnm{Guy}\binits{G.}},
  \bauthor{\bsnm{Lemey},~\bfnm{Philippe}\binits{P.}} \AND
  \bauthor{\bsnm{Suchard},~\bfnm{Marc~A.}\binits{M.~A.}}
(\byear{2016}).
\btitle{Genealogical Working Distributions for Bayesian Model Testing with
  Phylogenetic Uncertainty}.
\bjournal{Systematic Biology}
\bvolume{65}
\bpages{250-264}.
\bdoi{10.1093/sysbio/syv083}
\end{barticle}
\endbibitem

\bibitem[\protect\citeauthoryear{Baele et~al.}{2012}]{Baele2012}
\begin{barticle}[author]
\bauthor{\bsnm{Baele},~\bfnm{Guy}\binits{G.}},
  \bauthor{\bsnm{Lemey},~\bfnm{Philippe}\binits{P.}},
  \bauthor{\bsnm{Bedford},~\bfnm{Trevor}\binits{T.}},
  \bauthor{\bsnm{Rambaut},~\bfnm{Andrew}\binits{A.}},
  \bauthor{\bsnm{Suchard},~\bfnm{Marc~A.}\binits{M.~A.}} \AND
  \bauthor{\bsnm{Alekseyenko},~\bfnm{Alexander~V.}\binits{A.~V.}}
(\byear{2012}).
\btitle{{Improving the Accuracy of Demographic and Molecular Clock Model
  Comparison While Accommodating Phylogenetic Uncertainty}}.
\bjournal{Molecular Biology and Evolution}
\bvolume{29}
\bpages{2157-2167}.
\bdoi{10.1093/molbev/mss084}
\end{barticle}
\endbibitem

\bibitem[\protect\citeauthoryear{Baele et~al.}{2025}]{Baele2025}
\begin{barticle}[author]
\bauthor{\bsnm{Baele},~\bfnm{G.}\binits{G.}},
  \bauthor{\bsnm{Ji},~\bfnm{X.}\binits{X.}},
  \bauthor{\bsnm{Hassler},~\bfnm{G.~W.}\binits{G.~W.}},
  \bauthor{\bsnm{McCrone},~\bfnm{J.~T.}\binits{J.~T.}},
  \bauthor{\bsnm{Shao},~\bfnm{Y.}\binits{Y.}},
  \bauthor{\bsnm{Zhang},~\bfnm{Z.}\binits{Z.}},
  \bauthor{\bsnm{Holbrook},~\bfnm{A.~J.}\binits{A.~J.}},
  \bauthor{\bsnm{Lemey},~\bfnm{P.}\binits{P.}},
  \bauthor{\bsnm{Drummond},~\bfnm{A.~J.}\binits{A.~J.}},
  \bauthor{\bsnm{Rambaut},~\bfnm{A.}\binits{A.}} \AND
  \bauthor{\bsnm{Suchard},~\bfnm{M.~A.}\binits{M.~A.}}
(\byear{2025}).
\btitle{{BEAST X} for Bayesian phylogenetic, phylogeographic and phylodynamic
  inference}.
\bjournal{Nature Methods}
\bvolume{22}
\bpages{1653--1656}.
\bdoi{10.1038/s41592-025-02751-x}
\end{barticle}
\endbibitem

\bibitem[\protect\citeauthoryear{Bielejec et~al.}{2014a}]{Bielejec2014}
\begin{barticle}[author]
\bauthor{\bsnm{Bielejec},~\bfnm{Filip}\binits{F.}},
  \bauthor{\bsnm{Lemey},~\bfnm{Philippe}\binits{P.}},
  \bauthor{\bsnm{Baele},~\bfnm{Guy}\binits{G.}},
  \bauthor{\bsnm{Rambaut},~\bfnm{Andrew}\binits{A.}} \AND
  \bauthor{\bsnm{Suchard},~\bfnm{Marc~A.}\binits{M.~A.}}
(\byear{2014}a).
\btitle{Inferring Heterogeneous Evolutionary Processes Through Time: from
  Sequence Substitution to Phylogeography}.
\bjournal{Systematic Biology}
\bvolume{63}
\bpages{493-504}.
\bdoi{10.1093/sysbio/syu015}
\end{barticle}
\endbibitem

\bibitem[\protect\citeauthoryear{Bielejec et~al.}{2014b}]{pibuss}
\begin{barticle}[author]
\bauthor{\bsnm{Bielejec},~\bfnm{Filip}\binits{F.}},
  \bauthor{\bsnm{Lemey},~\bfnm{Philippe}\binits{P.}},
  \bauthor{\bsnm{Carvalho},~\bfnm{Luiz~M.}\binits{L.~M.}},
  \bauthor{\bsnm{Baele},~\bfnm{Guy}\binits{G.}},
  \bauthor{\bsnm{Rambaut},~\bfnm{Andrew}\binits{A.}} \AND
  \bauthor{\bsnm{Suchard},~\bfnm{Marc~A.}\binits{M.~A.}}
(\byear{2014}b).
\btitle{$\pi$ BUSS: a parallel BEAST/BEAGLE utility for sequence simulation
  under complex evolutionary scenarios}.
\bjournal{BMC Bioinformatics}
\bvolume{15}
\bpages{133}.
\end{barticle}
\endbibitem

\bibitem[\protect\citeauthoryear{Bowen et~al.}{2000}]{Bowen2000}
\begin{barticle}[author]
\bauthor{\bsnm{Bowen},~\bfnm{Michael~D.}\binits{M.~D.}},
  \bauthor{\bsnm{Rollin},~\bfnm{Pierre~E.}\binits{P.~E.}},
  \bauthor{\bsnm{Ksiazek},~\bfnm{Thomas~G.}\binits{T.~G.}},
  \bauthor{\bsnm{Hustad},~\bfnm{Heather~L.}\binits{H.~L.}},
  \bauthor{\bsnm{Bausch},~\bfnm{Daniel~G.}\binits{D.~G.}},
  \bauthor{\bsnm{Demby},~\bfnm{Austin~H.}\binits{A.~H.}},
  \bauthor{\bsnm{Bajani},~\bfnm{Mary~D.}\binits{M.~D.}},
  \bauthor{\bsnm{Peters},~\bfnm{Clarence~J.}\binits{C.~J.}} \AND
  \bauthor{\bsnm{Nichol},~\bfnm{Stuart~T.}\binits{S.~T.}}
(\byear{2000}).
\btitle{Genetic Diversity among Lassa Virus Strains}.
\bjournal{Journal of Virology}
\bvolume{74}
\bpages{6992-7004}.
\bdoi{10.1128/jvi.74.15.6992-7004.2000}
\end{barticle}
\endbibitem

\bibitem[\protect\citeauthoryear{Drummond and Suchard}{2010}]{Drummond2010}
\begin{barticle}[author]
\bauthor{\bsnm{Drummond},~\bfnm{Alexei~J.}\binits{A.~J.}} \AND
  \bauthor{\bsnm{Suchard},~\bfnm{Marc~A.}\binits{M.~A.}}
(\byear{2010}).
\btitle{Bayesian random local clocks, or one rate to rule them all}.
\bjournal{BMC Biology}
\bvolume{8}
\bpages{114}.
\bdoi{10.1186/1741-7007-8-114}
\end{barticle}
\endbibitem

\bibitem[\protect\citeauthoryear{Drummond et~al.}{2006}]{Drummond2006}
\begin{barticle}[author]
\bauthor{\bsnm{Drummond},~\bfnm{Alexei~J.}\binits{A.~J.}},
  \bauthor{\bsnm{Ho},~\bfnm{Simon Y.~W.}\binits{S.~Y.~W.}},
  \bauthor{\bsnm{Phillips},~\bfnm{Matthew~J.}\binits{M.~J.}} \AND
  \bauthor{\bsnm{Rambaut},~\bfnm{Andrew}\binits{A.}}
(\byear{2006}).
\btitle{Relaxed Phylogenetics and Dating with Confidence}.
\bjournal{PLoS Biology}
\bvolume{4}
\bpages{e88}.
\bdoi{10.1371/journal.pbio.0040088}
\end{barticle}
\endbibitem

\bibitem[\protect\citeauthoryear{Duch{\^e}ne, Holmes and
  Ho}{2014}]{Duchene2014}
\begin{barticle}[author]
\bauthor{\bsnm{Duch{\^e}ne},~\bfnm{Sebasti{\'a}n}\binits{S.}},
  \bauthor{\bsnm{Holmes},~\bfnm{Edward~C}\binits{E.~C.}} \AND
  \bauthor{\bsnm{Ho},~\bfnm{Simon~YW}\binits{S.~Y.}}
(\byear{2014}).
\btitle{Analyses of evolutionary dynamics in viruses are hindered by a
  time-dependent bias in rate estimates}.
\bjournal{Proceedings of the Royal Society B: Biological Sciences}
\bvolume{281}
\bpages{20140732}.
\end{barticle}
\endbibitem

\bibitem[\protect\citeauthoryear{Dudas et~al.}{2016}]{dudas2016}
\begin{barticle}[author]
\bauthor{\bsnm{Dudas},~\bfnm{Gytis}\binits{G.}},
  \bauthor{\bsnm{Carvalho},~\bfnm{Luiz~Max}\binits{L.~M.}},
  \bauthor{\bsnm{Bedford},~\bfnm{Trevor}\binits{T.}},
  \bauthor{\bsnm{Tatem},~\bfnm{Andrew~J}\binits{A.~J.}},
  \bauthor{\bsnm{Baele},~\bfnm{Guy}\binits{G.}},
  \bauthor{\bsnm{Faria},~\bfnm{Nuno~R}\binits{N.~R.}},
  \bauthor{\bsnm{Park},~\bfnm{Daniel~J}\binits{D.~J.}},
  \bauthor{\bsnm{Ladner},~\bfnm{Jason~T}\binits{J.~T.}},
  \bauthor{\bsnm{Arias},~\bfnm{Armando}\binits{A.}},
  \bauthor{\bsnm{Asogun},~\bfnm{Danny}\binits{D.}} \betal{et~al.}
(\byear{2016}).
\btitle{Virus genomes reveal factors that spread and sustained the Ebola
  epidemic}.
\bjournal{Nature}
\bvolume{544}
\bpages{309--315}.
\bdoi{10.1038/nature22040}
\end{barticle}
\endbibitem

\bibitem[\protect\citeauthoryear{Felsenstein}{1973}]{Felsenstein1973}
\begin{barticle}[author]
\bauthor{\bsnm{Felsenstein},~\bfnm{Joseph}\binits{J.}}
(\byear{1973}).
\btitle{Maximum Likelihood and Minimum-Steps Methods for Estimating
  Evolutionary Trees from Data on Discrete Characters}.
\bjournal{Systematic Biology}
\bvolume{22}
\bpages{240-249}.
\bdoi{10.1093/sysbio/22.3.240}
\end{barticle}
\endbibitem

\bibitem[\protect\citeauthoryear{Felsenstein}{1981}]{Felsenstein1981}
\begin{barticle}[author]
\bauthor{\bsnm{Felsenstein},~\bfnm{Joseph}\binits{J.}}
(\byear{1981}).
\btitle{Evolutionary trees from {DNA} sequences: A maximum likelihood
  approach}.
\bjournal{Journal of Molecular Evolution}
\bvolume{17}
\bpages{368-376}.
\bdoi{10.1007/BF01734359}
\end{barticle}
\endbibitem

\bibitem[\protect\citeauthoryear{Fortmann and Hitz}{1977}]{Fortmann1977}
\begin{bbook}[author]
\bauthor{\bsnm{Fortmann},~\bfnm{Thomas~E.}\binits{T.~E.}} \AND
  \bauthor{\bsnm{Hitz},~\bfnm{K.~L.}\binits{K.~L.}}
(\byear{1977}).
\btitle{An Introduction to Linear Control Systems},
\bedition{1st} ed.
\bpublisher{CRC Press}.
\end{bbook}
\endbibitem

\bibitem[\protect\citeauthoryear{Gelfand and Vounatsou}{2003}]{Gelfand2003}
\begin{barticle}[author]
\bauthor{\bsnm{Gelfand},~\bfnm{Alan~E.}\binits{A.~E.}} \AND
  \bauthor{\bsnm{Vounatsou},~\bfnm{Penelope}\binits{P.}}
(\byear{2003}).
\btitle{Proper multivariate conditional autoregressive models for spatial data
  analysis}.
\bjournal{Biostatistics}
\bvolume{4}
\bpages{11-15}.
\bdoi{10.1093/biostatistics/4.1.11}
\end{barticle}
\endbibitem

\bibitem[\protect\citeauthoryear{Gill et~al.}{2012}]{Skygrid}
\begin{barticle}[author]
\bauthor{\bsnm{Gill},~\bfnm{Mandev~S.}\binits{M.~S.}},
  \bauthor{\bsnm{Lemey},~\bfnm{Philippe}\binits{P.}},
  \bauthor{\bsnm{Faria},~\bfnm{Nuno~R.}\binits{N.~R.}},
  \bauthor{\bsnm{Rambaut},~\bfnm{Andrew}\binits{A.}},
  \bauthor{\bsnm{Shapiro},~\bfnm{Beth}\binits{B.}} \AND
  \bauthor{\bsnm{Suchard},~\bfnm{Marc~A.}\binits{M.~A.}}
(\byear{2012}).
\btitle{Improving Bayesian Population Dynamics Inference: A Coalescent-Based
  Model for Multiple Loci}.
\bjournal{Molecular Biology and Evolution}
\bvolume{30}
\bpages{713-724}.
\bdoi{10.1093/molbev/mss265}
\end{barticle}
\endbibitem

\bibitem[\protect\citeauthoryear{Gill et~al.}{2025}]{Gill:2025aa}
\begin{barticle}[author]
\bauthor{\bsnm{Gill},~\bfnm{Mandev~S}\binits{M.~S.}},
  \bauthor{\bsnm{Baele},~\bfnm{Guy}\binits{G.}},
  \bauthor{\bsnm{Suchard},~\bfnm{Marc~A}\binits{M.~A.}} \AND
  \bauthor{\bsnm{Lemey},~\bfnm{Philippe}\binits{P.}}
(\byear{2025}).
\btitle{Infinite Mixture Models for Improved Modeling of Across-Site
  Evolutionary Variation}.
\bjournal{Mol Biol Evol}
\bvolume{42}
\bpages{msaf199}.
\bdoi{10.1093/molbev/msaf199}
\end{barticle}
\endbibitem

\bibitem[\protect\citeauthoryear{Hasegawa, Kishino and Yano}{1985}]{hky}
\begin{barticle}[author]
\bauthor{\bsnm{Hasegawa},~\bfnm{Masami}\binits{M.}},
  \bauthor{\bsnm{Kishino},~\bfnm{Hirohisa}\binits{H.}} \AND
  \bauthor{\bsnm{Yano},~\bfnm{Taka-aki}\binits{T.-a.}}
(\byear{1985}).
\btitle{Dating of the human-ape splitting by a molecular clock of mitochondrial
  {DNA}}.
\bjournal{Journal of Molecular Evolution}
\bvolume{22}
\bpages{160--174}.
\end{barticle}
\endbibitem

\bibitem[\protect\citeauthoryear{Hasegawa, Yano and Kishino}{1984}]{hky1985}
\begin{barticle}[author]
\bauthor{\bsnm{Hasegawa},~\bfnm{M.}\binits{M.}},
  \bauthor{\bsnm{Yano},~\bfnm{T.}\binits{T.}} \AND
  \bauthor{\bsnm{Kishino},~\bfnm{H.}\binits{H.}}
(\byear{1984}).
\btitle{A new molecular clock of mitochondrial {DNA} and the evolution of
  {H}ominoids}.
\bjournal{Proceedings of the Japan Academy, Series B}
\bvolume{60}
\bpages{95--98}.
\end{barticle}
\endbibitem

\bibitem[\protect\citeauthoryear{Ji et~al.}{2019}]{Ji2019}
\begin{barticle}[author]
\bauthor{\bsnm{Ji},~\bfnm{Xiang}\binits{X.}},
  \bauthor{\bsnm{Zhang},~\bfnm{Zhenyu}\binits{Z.}},
  \bauthor{\bsnm{Holbrook},~\bfnm{Andrew}\binits{A.}},
  \bauthor{\bsnm{Nishimura},~\bfnm{Akihiko}\binits{A.}},
  \bauthor{\bsnm{Baele},~\bfnm{Guy}\binits{G.}},
  \bauthor{\bsnm{Rambaut},~\bfnm{Andrew}\binits{A.}},
  \bauthor{\bsnm{Lemey},~\bfnm{Philippe}\binits{P.}} \AND
  \bauthor{\bsnm{Suchard},~\bfnm{Marc~A}\binits{M.~A.}}
(\byear{2019}).
\btitle{{Gradients Do Grow on Trees: A Linear-Time O(N)-Dimensional Gradient
  for Statistical Phylogenetics}}.
\bjournal{Mol Biol Evol}
\bvolume{37}
\bpages{3047-3060}.
\end{barticle}
\endbibitem

\bibitem[\protect\citeauthoryear{Kailath}{1980}]{Kailath1980}
\begin{bbook}[author]
\bauthor{\bsnm{Kailath},~\bfnm{T}\binits{T.}}
(\byear{1980}).
\btitle{{Linear Systems}}.
\bpublisher{Prentice Hall}.
\end{bbook}
\endbibitem

\bibitem[\protect\citeauthoryear{Kishino, Thorne and Bruno}{2001}]{Kishino2001}
\begin{barticle}[author]
\bauthor{\bsnm{Kishino},~\bfnm{Hirohisa}\binits{H.}},
  \bauthor{\bsnm{Thorne},~\bfnm{Jeffrey~L.}\binits{J.~L.}} \AND
  \bauthor{\bsnm{Bruno},~\bfnm{William~J.}\binits{W.~J.}}
(\byear{2001}).
\btitle{{Performance of a Divergence Time Estimation Method under a
  Probabilistic Model of Rate Evolution}}.
\bjournal{Molecular Biology and Evolution}
\bvolume{18}
\bpages{352-361}.
\bdoi{10.1093/oxfordjournals.molbev.a003811}
\end{barticle}
\endbibitem

\bibitem[\protect\citeauthoryear{Lemey et~al.}{2014}]{Lemey2014}
\begin{barticle}[author]
\bauthor{\bsnm{Lemey},~\bfnm{Philippe}\binits{P.}},
  \bauthor{\bsnm{Rambaut},~\bfnm{Andrew}\binits{A.}},
  \bauthor{\bsnm{Bedford},~\bfnm{Trevor}\binits{T.}},
  \bauthor{\bsnm{Faria},~\bfnm{Nuno}\binits{N.}},
  \bauthor{\bsnm{Bielejec},~\bfnm{Filip}\binits{F.}},
  \bauthor{\bsnm{Baele},~\bfnm{Guy}\binits{G.}},
  \bauthor{\bsnm{Russell},~\bfnm{Colin~A.}\binits{C.~A.}},
  \bauthor{\bsnm{Smith},~\bfnm{Derek~J.}\binits{D.~J.}},
  \bauthor{\bsnm{Pybus},~\bfnm{Oliver~G.}\binits{O.~G.}},
  \bauthor{\bsnm{Brockmann},~\bfnm{Dirk}\binits{D.}} \AND
  \bauthor{\bsnm{Suchard},~\bfnm{Marc~A.}\binits{M.~A.}}
(\byear{2014}).
\btitle{Unifying Viral Genetics and Human Transportation Data to Predict the
  Global Transmission Dynamics of Human Influenza H3N2}.
\bjournal{PLOS Pathogens}
\bvolume{10}
\bpages{1-10}.
\bdoi{10.1371/journal.ppat.1003932}
\end{barticle}
\endbibitem

\bibitem[\protect\citeauthoryear{Lemey et~al.}{2021}]{Lemey2021}
\begin{barticle}[author]
\bauthor{\bsnm{Lemey},~\bfnm{Philippe}\binits{P.}},
  \bauthor{\bsnm{Ruktanonchai},~\bfnm{Nick}\binits{N.}},
  \bauthor{\bsnm{Hong},~\bfnm{Samuel}\binits{S.}},
  \bauthor{\bsnm{Colizza},~\bfnm{Vittoria}\binits{V.}},
  \bauthor{\bsnm{Poletto},~\bfnm{Chiara}\binits{C.}},
  \bauthor{\bsnm{Broeck},~\bfnm{Frederik}\binits{F.}},
  \bauthor{\bsnm{Gill},~\bfnm{Mandev}\binits{M.}},
  \bauthor{\bsnm{Ji},~\bfnm{Xiang}\binits{X.}},
  \bauthor{\bsnm{Levasseur},~\bfnm{Anthony}\binits{A.}},
  \bauthor{\bsnm{Oude~Munnink},~\bfnm{Bas}\binits{B.}},
  \bauthor{\bsnm{Koopmans},~\bfnm{Marion}\binits{M.}},
  \bauthor{\bsnm{Sadilek},~\bfnm{Adam}\binits{A.}},
  \bauthor{\bsnm{Lai},~\bfnm{Shengjie}\binits{S.}},
  \bauthor{\bsnm{Tatem},~\bfnm{Andrew}\binits{A.}},
  \bauthor{\bsnm{Baele},~\bfnm{Guy}\binits{G.}},
  \bauthor{\bsnm{Suchard},~\bfnm{Marc}\binits{M.}} \AND
  \bauthor{\bsnm{Dellicour},~\bfnm{Simon}\binits{S.}}
(\byear{2021}).
\btitle{Untangling introductions and persistence in COVID-19 resurgence in
  Europe}.
\bjournal{Nature}
\bvolume{595}.
\bdoi{10.1038/s41586-021-03754-2}
\end{barticle}
\endbibitem

\bibitem[\protect\citeauthoryear{Membrebe et~al.}{2019}]{Membrebe2019}
\begin{barticle}[author]
\bauthor{\bsnm{Membrebe},~\bfnm{Jade~Vincent}\binits{J.~V.}},
  \bauthor{\bsnm{Suchard},~\bfnm{Marc~A}\binits{M.~A.}},
  \bauthor{\bsnm{Rambaut},~\bfnm{Andrew}\binits{A.}},
  \bauthor{\bsnm{Baele},~\bfnm{Guy}\binits{G.}} \AND
  \bauthor{\bsnm{Lemey},~\bfnm{Philippe}\binits{P.}}
(\byear{2019}).
\btitle{{Bayesian Inference of Evolutionary Histories under Time-Dependent
  Substitution Rates}}.
\bjournal{Molecular Biology and Evolution}
\bvolume{36}
\bpages{1793-1803}.
\bdoi{10.1093/molbev/msz094}
\end{barticle}
\endbibitem

\bibitem[\protect\citeauthoryear{Minin, Bloomquist and Suchard}{2008}]{Skyride}
\begin{barticle}[author]
\bauthor{\bsnm{Minin},~\bfnm{Vladimir~N.}\binits{V.~N.}},
  \bauthor{\bsnm{Bloomquist},~\bfnm{Erik~W.}\binits{E.~W.}} \AND
  \bauthor{\bsnm{Suchard},~\bfnm{Marc~A.}\binits{M.~A.}}
(\byear{2008}).
\btitle{Smooth Skyride through a Rough Skyline: Bayesian Coalescent-Based
  Inference of Population Dynamics}.
\bjournal{Molecular Biology and Evolution}
\bvolume{25}
\bpages{1459-1471}.
\bdoi{10.1093/molbev/msn090}
\end{barticle}
\endbibitem

\bibitem[\protect\citeauthoryear{Neal}{2011}]{Neal2011}
\begin{bbook}[author]
\bauthor{\bsnm{Neal},~\bfnm{Radford~M.}\binits{R.~M.}}
(\byear{2011}).
\btitle{MCMC using Hamiltonian dynamics},
\bedition{1st} ed.
\bpublisher{Chapman and Hall/CRC}.
\end{bbook}
\endbibitem

\bibitem[\protect\citeauthoryear{Nunes et~al.}{2014}]{Nunes2014}
\begin{barticle}[author]
\bauthor{\bsnm{Nunes},~\bfnm{Marcio R.~T.}\binits{M.~R.~T.}},
  \bauthor{\bsnm{Palacios},~\bfnm{Gustavo}\binits{G.}},
  \bauthor{\bsnm{Faria},~\bfnm{Nuno~Rodrigues}\binits{N.~R.}},
  \bauthor{\bsnm{Sousa},~\bfnm{Edivaldo~Costa}\binits{E.~C.} \bsuffix{Jr}},
  \bauthor{\bsnm{Pantoja},~\bfnm{Jamilla~A.}\binits{J.~A.}},
  \bauthor{\bsnm{Rodrigues},~\bfnm{Sueli~G.}\binits{S.~G.}},
  \bauthor{\bsnm{Carvalho},~\bfnm{Val{\'e}ria~L.}\binits{V.~L.}},
  \bauthor{\bsnm{Medeiros},~\bfnm{Daniele B.~A.}\binits{D.~B.~A.}},
  \bauthor{\bsnm{Savji},~\bfnm{Nazir}\binits{N.}},
  \bauthor{\bsnm{Baele},~\bfnm{Guy}\binits{G.}},
  \bauthor{\bsnm{Suchard},~\bfnm{Marc~A.}\binits{M.~A.}},
  \bauthor{\bsnm{Lemey},~\bfnm{Philippe}\binits{P.}},
  \bauthor{\bsnm{Vasconcelos},~\bfnm{Pedro F.~C.}\binits{P.~F.~C.}} \AND
  \bauthor{\bsnm{Lipkin},~\bfnm{W.~Ian}\binits{W.~I.}}
(\byear{2014}).
\btitle{Air Travel Is Associated with Intracontinental Spread of Dengue Virus
  Serotypes 1--3 in Brazil}.
\bjournal{PLOS Neglected Tropical Diseases}
\bvolume{8}
\bpages{1-13}.
\bdoi{10.1371/journal.pntd.0002769}
\end{barticle}
\endbibitem

\bibitem[\protect\citeauthoryear{Pybus et~al.}{2012}]{Pybus2012}
\begin{barticle}[author]
\bauthor{\bsnm{Pybus},~\bfnm{Oliver~G.}\binits{O.~G.}},
  \bauthor{\bsnm{Suchard},~\bfnm{Marc~A.}\binits{M.~A.}},
  \bauthor{\bsnm{Lemey},~\bfnm{Philippe}\binits{P.}},
  \bauthor{\bsnm{Bernardin},~\bfnm{Flavien~J.}\binits{F.~J.}},
  \bauthor{\bsnm{Rambaut},~\bfnm{Andrew}\binits{A.}},
  \bauthor{\bsnm{Crawford},~\bfnm{Forrest~W.}\binits{F.~W.}},
  \bauthor{\bsnm{Gray},~\bfnm{Rebecca~R.}\binits{R.~R.}},
  \bauthor{\bsnm{Arinaminpathy},~\bfnm{Nimalan}\binits{N.}},
  \bauthor{\bsnm{Stramer},~\bfnm{Susan~L.}\binits{S.~L.}},
  \bauthor{\bsnm{Busch},~\bfnm{Michael~P.}\binits{M.~P.}} \AND
  \bauthor{\bsnm{Delwart},~\bfnm{Eric~L.}\binits{E.~L.}}
(\byear{2012}).
\btitle{Unifying the spatial epidemiology and molecular evolution of emerging
  epidemics}.
\bjournal{Proceedings of the National Academy of Sciences}
\bvolume{109}
\bpages{15066-15071}.
\bdoi{10.1073/pnas.1206598109}
\end{barticle}
\endbibitem

\bibitem[\protect\citeauthoryear{Rambaut}{2000}]{Rambaut2000}
\begin{barticle}[author]
\bauthor{\bsnm{Rambaut},~\bfnm{Andrew}\binits{A.}}
(\byear{2000}).
\btitle{Estimating the rate of molecular evolution: incorporating
  non-contemporaneous sequences into maximum likelihood phylogenies}.
\bjournal{Bioinformatics}
\bvolume{16}
\bpages{395-399}.
\bdoi{10.1093/bioinformatics/16.4.395}
\end{barticle}
\endbibitem

\bibitem[\protect\citeauthoryear{Rambaut et~al.}{2018}]{Tracer}
\begin{barticle}[author]
\bauthor{\bsnm{Rambaut},~\bfnm{Andrew}\binits{A.}},
  \bauthor{\bsnm{Drummond},~\bfnm{Alexei~J}\binits{A.~J.}},
  \bauthor{\bsnm{Xie},~\bfnm{Dong}\binits{D.}},
  \bauthor{\bsnm{Baele},~\bfnm{Guy}\binits{G.}} \AND
  \bauthor{\bsnm{Suchard},~\bfnm{Marc~A}\binits{M.~A.}}
(\byear{2018}).
\btitle{Posterior Summarization in Bayesian Phylogenetics Using Tracer 1.7}.
\bjournal{Systematic Biology}
\bvolume{67}
\bpages{901-904}.
\bdoi{10.1093/sysbio/syy032}
\end{barticle}
\endbibitem

\bibitem[\protect\citeauthoryear{Rasmussen and Williams}{2006}]{Rasmussen2006}
\begin{bbook}[author]
\bauthor{\bsnm{Rasmussen},~\bfnm{Carl~Edward}\binits{C.~E.}} \AND
  \bauthor{\bsnm{Williams},~\bfnm{Christopher K.~I.}\binits{C.~K.~I.}}
(\byear{2006}).
\btitle{Gaussian Processes for Machine Learning}.
\bpublisher{MIT Press}, \baddress{Cambridge, MA}.
\end{bbook}
\endbibitem

\bibitem[\protect\citeauthoryear{Rindos et~al.}{1995}]{Rindos}
\begin{barticle}[author]
\bauthor{\bsnm{Rindos},~\bfnm{A}\binits{A.}},
  \bauthor{\bsnm{Woolet},~\bfnm{S}\binits{S.}},
  \bauthor{\bsnm{Viniotis},~\bfnm{I}\binits{I.}} \AND
  \bauthor{\bsnm{Trivedi},~\bfnm{K.~S.}\binits{K.~S.}}
(\byear{1995}).
\btitle{{Exact methods for the transient analysis for non-homogeneous
  continuous-time Markov chains}}.
\bjournal{Numerical Solutions of Markov Chains(NSMC)}
\bpages{121-134}.
\end{barticle}
\endbibitem

\bibitem[\protect\citeauthoryear{Rue and Held}{2005}]{rueheld}
\begin{bbook}[author]
\bauthor{\bsnm{Rue},~\bfnm{H{\aa}vard}\binits{H.}} \AND
  \bauthor{\bsnm{Held},~\bfnm{Leonard}\binits{L.}}
(\byear{2005}).
\btitle{Gaussian Markov Random Fields: Theory and Applications},
\bedition{1st} ed.
\bpublisher{Chapman and Hall/CRC}.
\bdoi{10.1201/9780203492024}
\end{bbook}
\endbibitem

\bibitem[\protect\citeauthoryear{Tavar\'e}{1986}]{Tavare1986}
\begin{barticle}[author]
\bauthor{\bsnm{Tavar\'e},~\bfnm{S.}\binits{S.}}
(\byear{1986}).
\btitle{Some probabilistic and statistical problems on the analysis of {DNA}
  sequences}.
\bjournal{Lectures on Mathematics in the Life Sciences}
\bvolume{17}
\bpages{57-86}.
\end{barticle}
\endbibitem

\bibitem[\protect\citeauthoryear{Yang and Rannala}{2005}]{YangRannala2005}
\begin{barticle}[author]
\bauthor{\bsnm{Yang},~\bfnm{Ziheng}\binits{Z.}} \AND
  \bauthor{\bsnm{Rannala},~\bfnm{Bruce}\binits{B.}}
(\byear{2005}).
\btitle{Bayesian Estimation of Species Divergence Times Under a Molecular Clock
  Using Multiple Fossil Calibrations with Soft Bounds}.
\bjournal{Molecular Biology and Evolution}
\bvolume{23}
\bpages{212-226}.
\bdoi{10.1093/molbev/msj024}
\end{barticle}
\endbibitem

\bibitem[\protect\citeauthoryear{Yoder and Yang}{2000}]{Yoder2000}
\begin{barticle}[author]
\bauthor{\bsnm{Yoder},~\bfnm{Anne~D.}\binits{A.~D.}} \AND
  \bauthor{\bsnm{Yang},~\bfnm{Ziheng}\binits{Z.}}
(\byear{2000}).
\btitle{{Estimation of Primate Speciation Dates Using Local Molecular Clocks}}.
\bjournal{Molecular Biology and Evolution}
\bvolume{17}
\bpages{1081-1090}.
\bdoi{10.1093/oxfordjournals.molbev.a026389}
\end{barticle}
\endbibitem

\bibitem[\protect\citeauthoryear{Yu et~al.}{2017}]{ggtree}
\begin{barticle}[author]
\bauthor{\bsnm{Yu},~\bfnm{Guangchuang}\binits{G.}},
  \bauthor{\bsnm{Smith},~\bfnm{David}\binits{D.}},
  \bauthor{\bsnm{Zhu},~\bfnm{Huachen}\binits{H.}},
  \bauthor{\bsnm{Guan},~\bfnm{Yi}\binits{Y.}} \AND
  \bauthor{\bsnm{Lam},~\bfnm{Tommy}\binits{T.}}
(\byear{2017}).
\btitle{ggtree : an R package for visualization and annotation of phylogenetic
  trees with their covariates and other associated data}.
\bjournal{Methods in Ecology and Evolution}
\bvolume{8}
\bpages{28-36}.
\bdoi{10.1111/2041-210X.12628}
\end{barticle}
\endbibitem

\bibitem[\protect\citeauthoryear{Zuckerkandl and
  Pauling}{1965}]{Zuckerland1965}
\begin{barticle}[author]
\bauthor{\bsnm{Zuckerkandl},~\bfnm{Emile}\binits{E.}} \AND
  \bauthor{\bsnm{Pauling},~\bfnm{Linus}\binits{L.}}
(\byear{1965}).
\btitle{Molecules as documents of evolutionary history}.
\bjournal{Journal of Theoretical Biology}
\bvolume{8}
\bpages{357-366}.
\bdoi{https://doi.org/10.1016/0022-5193(65)90083-4}
\end{barticle}
\endbibitem

\end{thebibliography}

\appendix

\section{Integrability of the polyepoch clock model posterior distribution}\label{sec:improperGMRF_appendix}

In this Supplementary Section, we demonstrate that the integral of the polyepoch clock model posterior distribution under the improper Gaussian Markov random field (GMRF) prior is not finite.

\begin{theorem}\label{thm:theorem1}
Let $\mathrm{P}(\boldsymbol{\theta} | \tau)$ denote the prior induced on the polyepoch clock model rate parameters $\boldsymbol{\theta} = (\theta_1, \dots, \theta_{M + 1})'$ by assigning an improper GMRF prior on $\log \boldsymbol{\theta} = (\log \theta_1, \dots, \log \theta_{M + 1})'$. 
Then, the integral of $\mathrm{P}(\boldsymbol{\theta} | \tau)$ over $[L, \infty)^{M + 1}$ is not finite for any $L > 0$.
\end{theorem}

\begin{proof}
Assigning the improper GMRF prior on $\log \boldsymbol{\theta}$ induces a prior on $\boldsymbol{\theta}$ that corresponds to the following  singular multivariate log-normal distribution:
\begin{equation}\label{eq:rate_improper_prior}
\mathrm{P}(\boldsymbol{\theta} \given \tau ) = \frac{(2 \pi)^{-M/2} \tau^M}{\prod_{m = 1}^{M + 1} \theta_m} \exp \left(-\frac{\tau}{2}  \left(\log \boldsymbol{\theta}\right)' \left(\mathbf{D_w}  - \mathbf{W} \right) \log \boldsymbol{\theta}  \right) ,
\end{equation}
where $\mathbf{D_w}$ and $\mathbf{W}$ are defined in the Bayesian Inference section of the main article. 
We first transform $\theta_m$ to $z_m = \log \theta_m$ for $m = 1, \dots, M + 1$. 
We denote the integral of $\mathrm{P}(\boldsymbol{\theta} \given \tau)$ over $[L, \infty)^{M + 1}$ with $I_{M+1}$, which can be written as 
\begin{eqnarray}\label{eq:proof_right_int}
I_{M + 1} & = & \int_{L}^{\infty} \dots \int_L^{\infty} \mathrm{P}(\boldsymbol{\theta}\given \tau) d \theta_{M + 1} \dots d \theta_{1} \nonumber\\ 
& = & \int_{\log L}^{\infty} \dots \int_{\log L}^{\infty} (2 \pi)^{-M/2} \tau^{M} \exp \left[-\frac{\tau}{2} \mathbf{z}' (\mathbf{D}_w - \mathbf{W}) \mathbf{z}\right] dz_{M + 1} \dots dz_{1},
\end{eqnarray}
where $\mathbf{z} = (z_1, \dots, z_{M + 1})'$. To accomplish this, we first show that $I_2$ is not finite.
Specifically,
\begin{eqnarray}\label{eq:proof_right_int2}
I_2
& = & \int_{\log L}^{\infty} \int_{\log L}^{\infty} \frac{\tau}{\sqrt{2 \pi}} \exp \left[-\frac{\tau}{2} (z_2 - z_1)^2\right] dz_2 dz_1  \nonumber\\
& = & \int_{\log L}^{\infty} 1 - \Phi(\sqrt{\tau}(\log L - z_1)) dz_1\\
& = & \geq \int_{\log L}^{\infty}  \frac{1}{2}  \ d z_1 = \infty.  \nonumber
\end{eqnarray}
We integrate with respect to $z_2$ first and then to get the inequality in the last step, we use the fact that  $\log L \leq z_1 <  \infty \implies 1/2 \leq 1 - \Phi(\sqrt{\tau}(\log L - z_1)) < 1$, where $\Phi(\cdot)$ is the cumulative distribution function of the standard normal distribution. 
Next, we show that if $I_M$ is not finite, then $I_{M + 1}$ is not finite, where $M \geq 2$ is an integer. 
We start by assuming $I_M$ is not finite. 
Let $K_{M}$ be defined as $K_{M} = (2\pi)^{-(M - 1)/2} \tau^{(M - 1)/2}$ . 
We consider the $M + 1$ dimensional integral $I_{M + 1}$ as shown below and integrate with respect to $z_{M + 1}$ first:
\begin{eqnarray}\label{eq:induction_posterior}
 I_{M + 1} 
& = & \int_{\log L}^{\infty} \dots \int_{\log L}^{\infty} K_{M + 1} \exp [-\frac{\tau}{2}\sum_{m = 1}^{M} (z_{m +1} - z_m)^2] dz_{M + 1} \dots dz_1  \nonumber\\
& = & \int_{\log L}^{\infty} \dots \int_{\log L}^{\infty} K_{M + 1} \exp [-\frac{\tau}{2}\sum_{m = 1}^{M-1} (z_{m +1} - z_m)^2]  \nonumber\\
&  & \int_{\log L}^{\infty} \exp [-\frac{\tau}{2} (z_{M + 1} - z_M)^2] dz_{M + 1} dz_M \dots dz_1 \\
& = & \int_{\log L}^{\infty} \dots \int_{\log L}^{\infty} K_{M} \exp [-\frac{\tau}{2}\sum_{m = 1}^{M-1} (z_{m +1} - z_m)^2] \nonumber\\ 
& & (1 - \Phi(\sqrt{\tau}(\log L - z_m))) dz_M \dots dz_1  \nonumber\\
& = & \geq \frac{1}{2}  \int_{\log L}^{\infty} \dots \int_{\log L}^{\infty} K_{M} \exp [-\frac{\tau}{2}\sum_{m = 1}^{M-1} (z_{m +1} - z_m)^2] dz_M \dots dz_1  \nonumber\\
& = & \frac{1}{2} I_M .  \nonumber
\end{eqnarray}
The inequality in the second last step follows from the fact that for $\log L \leq z_M < \infty$, we have $1/2 \leq 1 - \Phi(\sqrt{\tau}(\log L - z_M)) < 1$ and also since the remaining integrand is positive. 
Hence, we get $I_{M + 1} \geq \frac{1}{2}I_M$ and since we assumed $I_M$ is not finite, we have that $I_{M + 1}$ is not finite as well. 
Therefore by mathematical induction, we prove that $I_M$ is not finite for any integer $M \geq 2$. 
Hence, there does not exist any $L > 0$ such that the integral of $\mathrm{P}(\boldsymbol{\theta} | \tau)$ over $[L, \infty)^{M + 1}$ is finite for any $M \geq 1$.
\end{proof}
\begin{theorem}\label{thm:theorem2}
For a known time-calibrated phylogeny $\mathcal{F}$ where the sequences  $\mathbf{Y}^*$ for all nodes are known, let $\mathrm{P}(\boldsymbol{\theta} | \mathbf{Y}^*)$ denote the posterior distribution of $\boldsymbol{\theta}$ given $\mathbf{Y}^*$ under the polyepoch clock model. 
Then, assigning the improper GMRF prior on $\log \boldsymbol{\theta}$ yields a posterior distribution $\mathrm{P}(\boldsymbol{\theta} | \mathbf{Y}^*)$ that is improper.
\end{theorem}
\begin{proof}
For a known time-calibrated phylogeny $\mathcal{F}$ with $N$ tip nodes, the real time $t_i$ associated with node $i$ is known for $i = 1, \dots, 2N - 1$. 
We assume there are $C$ sites in the sequence alignment. 
At site $c$, 
we denote the character observed for the node $i$ with $y_i^{(c)*}$ and assume that the characters $y_i^{(c)*}$ are known for $i = 1, \dots, 2N - 1$ and $c = 1, \dots, C$. 
The likelihood of the  observed characters $\mathbf{y}^{(c)*}$ for all nodes at site $c$, given the tree $\mathcal{F}$ and rate parameters $\boldsymbol{\theta}$, under the polyepoch clock model with $M + 1$ epochs is
\begin{equation}\label{eq:proof_likelihood}
\mathrm{P}(\mathbf{y}^{(c)^*} | \boldsymbol{\theta}) = \mathrm{P}(y_{2N - 1}^{(c)*}) \prod_{i = 1}^{2N - 2} \mathrm{P}(y_i^{(c)*} | y_{\pa{i}}^{(c)*}, \boldsymbol{\theta}).
\end{equation}
We consider a single branch connecting node $i$ to $\pa{i}$ such that $w_p$ is the largest grid-point less than or equal to $t_i$ and $w_q$ is the smallest grid-point greater than or equal to $t_{\pa{i}}$, as shown in Figure~\ref{fig:post_proof_branch1}. 
We denote the expected number of substitutions along the branch connecting node $i$ to $\pa{i}$ with $b_i(\boldsymbol{\theta})$, where
\begin{equation}\label{eq:branchlength}
b_i(\boldsymbol{\theta}) =  \left( \theta_p (t_i - w_p) + \sum_{m = p} ^ {q - 1} \theta_{m + 1} (w_m - w_{m +1}) + \theta_{q + 1} (w_q - t_{\pa{i}} ) \right).
\end{equation} 
The probability $\mathrm{P}(y_i^{(c)*} | y_{\pa{i}}^{(c)*}, \boldsymbol{\theta})$ is the $(y_{\pa{i}}^{(c)*}, y_i^{(c)*})$-th element of the transition probability matrix 
$\mathbf{P}(t_{\pa{i}}, t_i) = \exp\left[ b_i(\boldsymbol{\theta}) \mathbf{Q}\right]$. Consider the spectral decomposition $\mathbf{Q} = \mathbf{B} \boldsymbol{\Lambda} \mathbf{B}^{-1}$, where $\boldsymbol{\Lambda}$ is the diagonal matrix of eigenvalues $0 = \lambda_1 > \lambda_2 \geq \dots \geq \lambda_S$. 
The columns of $\mathbf{B} = (r_{ik})$ and the rows of $\mathbf{B}^{-1} = (c_{kj})$ are the right and left eigenvectors respectively of $\mathbf{Q}$. 
The right eigenvector corresponding to the largest eigenvalue $0$ is $(1, \dots, 1)'$ and the corresponding left eigenvector is the stationary distribution $(\pi_1, \dots, \pi_s)$. 
Under this spectral decomposition, we can write
\begin{eqnarray}\label{eq:proof_lik_branch}
\mathrm{P}(y_i^{(c)*} | y_{\pa{i}}^{(c)*}, \boldsymbol{\theta}) & = & \sum_{l = 1}^{S} r_{y^{(c)*}_{\pa{i}} l} \exp \left[\lambda_l b_i(\boldsymbol{\theta})\right] c_{l y_i^{(c)*}}  \nonumber\\
& = & \pi_{y^{(c)*}_i} + \sum_{l = 2}^{S} r_{y^{(c)*}_{\pa{i}} l} \exp \left[\lambda_l b_i(\boldsymbol{\theta})\right] c_{l y_i^{(c)*}} ,
\end{eqnarray}
for $i = 1, \dots, 2N - 2$ and $c = 1, \dots, C$. 
Substituting $\mathrm{P}(y_i^{(c)*} | y_{\pa{i}}^{(c)*}, \boldsymbol{\theta})$ in the likelihood of the observed characters at site $c$ \eqref{eq:proof_likelihood}, we have 
\begin{equation}\label{eq:proof_likelihood2}
\mathrm{P}(\mathbf{y}^{(c)^*} | \boldsymbol{\theta}) = \pi_{y^{(c)*}_{2N - 1}} \prod_{i = 1}^{2N - 2} \left(\pi_{y^{(c)*}_i} + \sum_{l = 2}^{s} r_{y_{\pa{i}}^{(c)*}l} \exp \left[\lambda_l b_i(\boldsymbol{\theta})\right] c_{l y_i^{(c)*}}\right).
\end{equation}
The likelihood $\mathrm{P}(\mathbf{Y}^* | \boldsymbol{\theta})$ of the observed sequences for all the nodes is the product of the likelihood of the observed characters for all the nodes at all sites which is given by: 
\begin{eqnarray}
\mathrm{P}( \mathbf{Y}^* | \boldsymbol{\theta}) = \prod_{c = 1}^{C} \left[   \pi_{y^{(c)*}_{2N - 1}} \prod_{i = 1}^{2N - 2}  \left(\pi_{y^{(c)*}_i} +  \sum_{l = 2}^{S} r_{y_{\pa{i}}^{(c)*}l} \exp \left[\lambda_l b_i(\boldsymbol{\theta})\right] c_{l y_i^{(c)}} \right) \right] .
\end{eqnarray}
Using the improper GMRF prior for $\log \boldsymbol{\theta}$ yields the following posterior distribution:
\begin{eqnarray}\label{eq:full_post}
\mathrm{P}( \boldsymbol{\theta} \given \mathbf{Y}^*) & \propto & \mathrm{P}(\mathbf{Y}^* | \boldsymbol{\theta}) \mathrm{P}(\boldsymbol{\theta} | \tau) \nonumber\\
& = & \prod_{c = 1}^{C} \left[ \pi_{y^{(c)*}_{2N - 1}} \prod_{i = 1}^{2N - 2}  \left(\pi_{y^{(c)*}_i} +  \sum_{l = 2}^{s} r_{y_{\pa{i}}^{(c)*}l} \exp \left[\lambda_l b_i(\boldsymbol{\theta})\right] c_{l y_i^{(c)*}} \right) \right] \mathrm{P}(\boldsymbol{\theta} | \tau) .
\end{eqnarray}
Since $\mathrm{P}(\boldsymbol{\theta} | \mathbf{Y}^*)$ is positive, to prove that $\mathrm{P}(\boldsymbol{\theta} | \mathbf{Y}^*)$ is improper, it suffices to show that the integral of $\mathrm{P}(\boldsymbol{\theta} | \mathbf{Y}^*)$ is not finite over $[L, \infty)^{M + 1}$ for some $L > 0$. 
We first note that $\mathrm{P}(\boldsymbol{\theta} | \mathbf{Y}^*)/\mathrm{P}(\boldsymbol{\theta} | \tau)$ in Equation (\ref{eq:full_post}) can be expressed as a sum of products where the only term that is independent of $\boldsymbol{\theta}$ is $\prod_{c = 1}^{C}  \prod_{i = 1}^{2N - 1} \pi_{y^{(c)*}_i}$. 
Every other term in the sum is a product of constants independent of $\boldsymbol{\theta}$ and expressions of the form $\exp\left[\lambda_l b_i(\boldsymbol{\theta})\right]$ for some $l \in \{2, \dots, S\}$ and some $i \in \{1, \dots, 2N - 2\}$, and since $\lambda_l < 0$ for $l = 2, \dots, S$, $\mathrm{P}(\boldsymbol{\theta} \given \mathbf{Y})$ is $\mathcal{O}(\mathrm{P}(\boldsymbol{\theta} \given \tau))$ as $\boldsymbol{\theta} \rightarrow \boldsymbol{\infty}$. 
Therefore, as a consequence of Theorem~\ref{thm:theorem1},  
 the integral of  $\mathrm{P}(\boldsymbol{\theta} | \mathbf{Y}^*)$ \eqref{eq:full_post} over $[L, \infty)^{M + 1}$ is not finite for sufficiently large $L$ and hence $\mathrm{P}(\boldsymbol{\theta} | \mathbf{Y}^*)$ is improper.
\end{proof}
\begin{figure}
\centering
\includegraphics[width=0.8\linewidth]{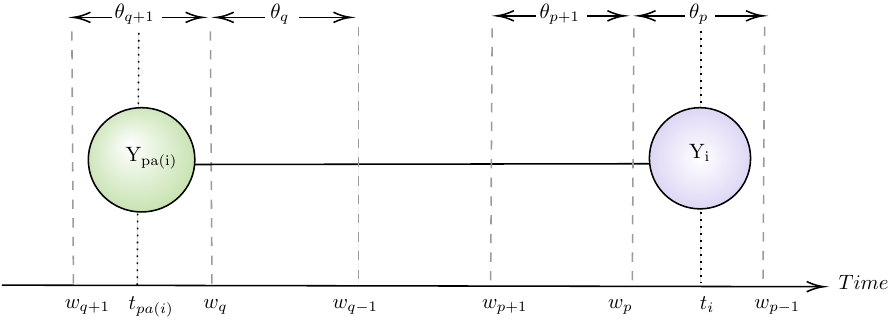} 
\caption{Illustration of a single branch connecting node $i$ with real time $t_i$ to it's parent node $\pa{i}$ with real time $t_{\pa{i}}$ under the polyepoch clock model, where $w_p$ is largest grid-point less than or equal to $t_i$ and $w_q$ is smallest grid-point greater than or equal to $t_{\pa{i}}$}
\label{fig:post_proof_branch1}
\end{figure}
\begin{corollary}
For a known time-calibrated phylogeny $\mathcal{F}$ under the polyepoch clock   model, where the sequences $\mathbf{Y}$ are observed at the tip nodes only, the improper GMRF prior on $\log \boldsymbol{\theta}$ yields an improper posterior distribution $\mathrm{P}(\boldsymbol{\theta} | \mathbf{Y})$.
\end{corollary}

\begin{proof}
Given the known time-calibrated phylogeny $\mathcal{F}$ and rates $\boldsymbol{\theta}$, the likelihood of the observed characters $\mathbf{y}^{(c)}$ at the tip nodes at site $c$ is the marginal probability of the observed characters at the tip nodes, marginalized over all possible latent states at the internal nodes at that site, which can be written as:
\begin{equation}
\mathrm{P}(\mathbf{y}^{(c)} | \boldsymbol{\theta}) = \sum_{y_{N + 1}^{(c)}} \dots \sum_{y^{(c)}_{2N - 1}} \left( \mathrm{P}(y^{(c)}_{2N - 1}) \prod_{i = 1}^{2N - 2} \mathrm{P}(y_i^{(c)} | y_{\pa{i}}^{(c)}, \boldsymbol{\theta}) \right).
\end{equation}
Consequently, the likelihood of the observed sequences at the tip nodes is the product of the likelihood of the observed characters at the tip nodes for all the $C$ sites as follows:
\begin{equation}\label{eq:cor1}
\mathrm{P}(\mathbf{Y} | \boldsymbol{\theta}) = \prod_{c = 1}^{C} \left[ \sum_{y_{N + 1}^{(c)}} \dots \sum_{y^{(c)}_{2N - 1}} \left( \mathrm{P}(y^{(c)}_{2N - 1}) \prod_{i = 1}^{2N - 2} \mathrm{P}(y_i^{(c)} | y_{\pa{i}}^{(c)}, \boldsymbol{\theta}) \right) \right] .
\end{equation}
The posterior distribution $\mathrm{P}(\boldsymbol{\theta} | \mathbf{Y})$ is
\begin{equation}\label{eq:corollary_proof}
\mathrm{P}(\boldsymbol{\theta} | \mathbf{Y}) =  \prod_{c = 1}^{C} \left[ \sum_{y_{N + 1}^{(c)}} \dots \sum_{y^{(c)}_{2N - 1}} \left( \mathrm{P}(y^{(c)}_{2N - 1}) \prod_{i = 1}^{2N - 2} \mathrm{P}(y_i^{(c)} | y_{\pa{i}}^{(c)}, \boldsymbol{\theta}) \right) \right] \mathrm{P}(\boldsymbol{\theta} | \tau) .
\end{equation}
The posterior above can be written as a sum of products, where each product corresponds to the posterior distribution of $\boldsymbol{\theta}$ given $\mathbf{Y}^*$, for the special case described in Theorem~\ref{thm:theorem2}, where the sequences $\mathbf{Y}^*$ are known for all the nodes of the known phylogeny. 
Hence, using Theorem~\ref{thm:theorem2}, we get that the integral of each product in the sum of products that $\mathrm{P}(\boldsymbol{\theta} | \mathbf{Y})$ in Equation \ref{eq:corollary_proof} breaks down into is infinite and hence $\mathrm{P}(\boldsymbol{\theta} | \mathbf{Y})$ \eqref{eq:corollary_proof}  is also improper when the sequences are observed for tip nodes only for a known phylogeny.
\end{proof}

\section{Sensitivity Analysis}\label{sec:sensitivity_appendix}
	
We assess the sensitivity of the results under the polyepoch clock model for different choices of the prior assigned to the GMRF precision $\tau$. 
The GMRF precision controls the smoothness of the estimated rate trajectory. 
Since, we do not know anything apriori about the smoothness, we choose a relatively uninformative prior for the precison. 
We fit the polyepoch clock model with 100 epochs to the simulated data with log linear evolutionary rate as described in Section 3.1 of the main article using the Gamma prior for the GMRF precision $\tau$. 
We fix the scale parameter to $1000$ and choose the shape parameter values $0.001$, $0.01$ and $0.1$. 
The resulting prior means are $1$, $10$ and $100$ and the prior resulting variances are $10^3$, $10^4$ and $10^5$ respectively. 
Figure~\ref{fig:sensitivity_analysis} shows that posterior boxplot of the log transformed GMRF precision, $\log \tau$ for the three different priors on the left. 
We observe that the posterior distribution of $\log \tau$ does not change significantly for the different choices of the prior. 
The plot on the right in Figure~\ref{fig:sensitivity_analysis} shows the posterior histogram (in blue) of $\log \tau$ against the flat and uninformative prior density (black dashed line), which suggests that the data is sufficiently informative to estimate the smoothness of the evolutionary rate trajectory.
	
\begin{figure}
\centering
\includegraphics[width=0.9\linewidth]{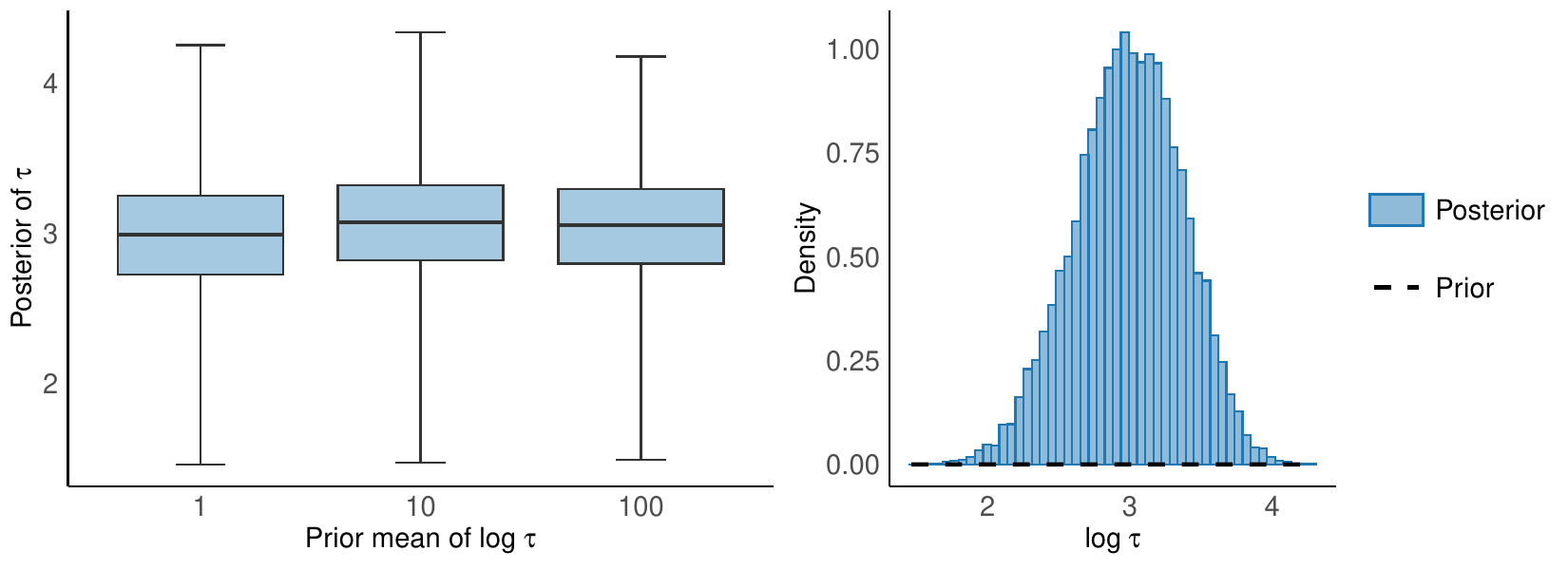}
\caption{The plot on the left shows the posterior boxplot of the log transformed GMRF precision $\tau$ for different Gamma priors with mean $1, 10$ and $100$ corresponding to a fixed scale $ = 1000$ and shape $0.001$, $0.01$ and $0.1$ from left to right. 
The plot on the right shows the posterior histogram and prior density of $\log \tau$ corresponding to Gamma$(0.001, 1000)$}.
\label{fig:sensitivity_analysis}
\end{figure}
    
 \section{Lassa Virus}\label{sec:Lassa_appendix}
 
 Lassa virus is a single-stranded RNA virus that can cause hemorrhagic fever with case fatality rates over $50\%$ and numerous sub-clinical infections among hospitalized Lassa Fever patients \citep{Andersen2015}. 
 Lassa fever is endemic to West Africa, causing tens of thousands hospitalizations and several thousand deaths each year. 
 The rodent Mastomys natalensis is the primary reservoir for the virus. 
 We analyze the S segment of the LASV genome dataset of \cite{Andersen2015} that consists of 211 samples obtained at clinics in both Sierra Leone and Nigeria, rodents in the field, laboratory isolates and previously sequenced genomes from 1969 to 2013.  
 We fit the polyepoch clock   model with 100 epochs of width 1 year between 1913 and 2013, 40 epochs of width 5 years between 1913 and 1713, and 20 epochs of width 50 years from 1713 to 713. 
 We choose denser grid-points during the recent years where all sampling times are concentrated and where we expect to capture a finer scale variation in the rate.
 Since all the samples are collected within 40 years from the most recent sampling time in 2013 and the TMRCA is known to be more than a thousand years \citep{Andersen2015}, there is not enough information available in the data to capture more precise variation in the rate as we move closer to the root.
 Therefore, we chose wider epochs beyond the time frame where the data are informative.
 The bottom left panel of Figure~\ref{fig:lassa} shows the posterior median (solid blue line) along with the $95\%$ BCI (blue shaded area) of the evolutionary rate from 713 to 2013. Due to the high uncertainty of the estimates from 613 to 1969, we cannot infer precisely whether there is any variation in the rate during this time period. 
 However, we observe that from 1969 to 2013 (dashed box), when all the samples were collected, there is significant variation in the evolutionary rate, which is magnified and shown in the top left panel of Figure~\ref{fig:lassa}. 
 The estimated rate increases by almost a factor of 5 from 1985 to 2006 and then drops by almost a factor of 4 from 2006 to 2013. 
 The top right panel of Figure~\ref{fig:lassa} shows the maximum clade credibility tree of the Lassa virus. 
 We identify the four major lineages of the Lassa virus, three in Nigeria (lineages I, II and III) and the fourth in Sierra Leone, Guinea and Liberia (lineage IV), which aligns with previous studies \citep{Bowen2000, Andersen2015}. 
 Our findings align with previous studies \citep{Andersen2015,  Ji2019} in estimating that the Lassa virus originated in modern-day Nigeria more than a thousand years ago and spread into the neighboring West African countries in the last few hundred years.
\begin{figure}
     \centering
   	\includegraphics[width=0.8\linewidth]{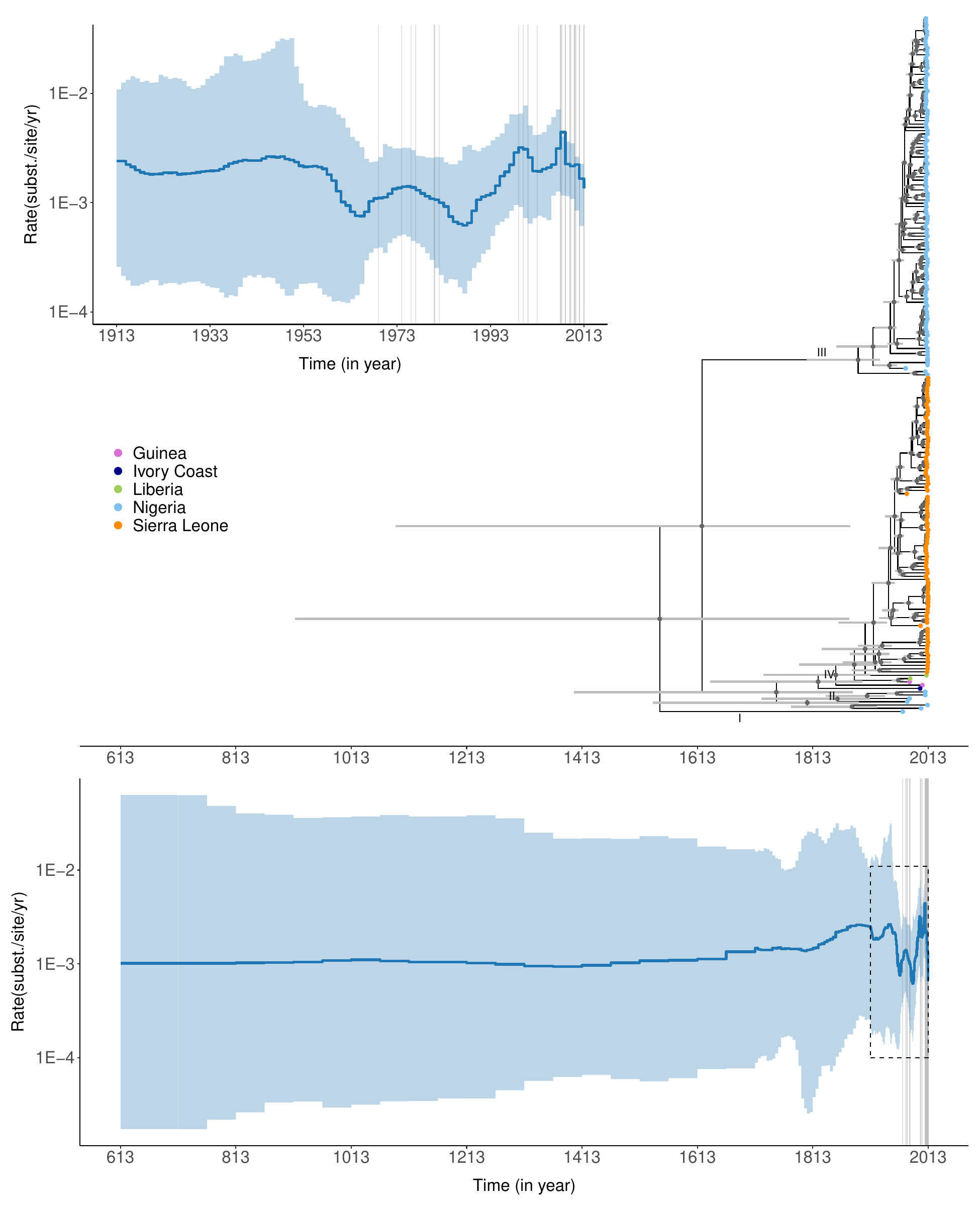} 
  	\caption{Evolutionary dynamics of the Lassa virus in West Africa from 13 to 2013 under the polyepoch clock model.
  	The top panel  shows the MCC tree, where the tips are colored according to geographical location and the posterior median and $95\%$ BCI of internal node dates are depicted with gray circles and gray shaded bars respectively.
  	The bottom panel shows the posterior median (solid blue line) and the $95\%$ BCI (blue shaded area) of the evolutionary rate. 
  	The vertical gray lines represent sampling times. 
  	The black dashed box showing the posterior median and $95\%$ BCI of the evolutionary rate from 1913 to 2013 is magnified in the top left panel.}
  	\label{fig:lassa}
\end{figure}
%%  imsart-nameyear.bst  will be used to                   %%
%%  MR numbers will be added by VTeX.                      %%
%%                                                         %%
%%  Use \cite{...} to cite references in text.             %%
%%                                                         %%
%%%%%%%%%%%%%%%%%%%%%%%%%%%%%%%%%%%%%%%%%%%%%%%%%%%%%%%%%%%%%

%% if your bibliography is in bibtex format, uncomment commands:
%\bibliographystyle{imsart-nameyear} % Style BST file
%\bibliography{bibliography}       % Bibliography file (usually '*.bib')

%% or include bibliography directly:
%\bibliographystyle{imsart-nameyear} % Style BST file
%\bibliography{refs}

\end{document}